\documentclass[hidelinks,onefignum,onetabnum,final]{siamart220329}

\usepackage{lipsum}
\usepackage{amsfonts}
\usepackage{graphicx}
\usepackage{epstopdf}
\usepackage{enumitem}
\usepackage{amsmath} 
\usepackage{amssymb}
\usepackage{stmaryrd}
\usepackage{mathtools}
\usepackage{mathrsfs}
\usepackage{mdframed}
\ifpdf
  \DeclareGraphicsExtensions{.eps,.pdf,.png,.jpg}
\else
  \DeclareGraphicsExtensions{.eps}
\fi

\newsiamremark{remark}{Remark}
\newsiamremark{hypothesis}{Hypothesis}
\newsiamremark{example}{Example}
\crefname{hypothesis}{Hypothesis}{Hypotheses}
\newsiamthm{claim}{Claim}

\title{Spectral Analysis of the Schrödinger Operator for the Incommensurate System\thanks{Submitted to the editors DATE.
\funding{This work was supported by the National Natural Science Foundation of China under grant  12571446, and the National Key R \& D Program of China under grants 2025YFA1016600 and 2025YFA1016601. }}}

\author{Yan Li\thanks{SKLMS, Academy of Mathematics and Systems Science, Chinese Academy of Sciences, Beijing 100190, China; and School of Mathematical Sciences, University of Chinese Academy of Sciences, Beijing 100049, China (\email{liyan2021@lsec.cc.ac.cn}, \email{azhou@lsec.cc.ac.cn}).} 
\and Yujian Song\thanks{School of Mathematical Sciences, Peking University, Beijing 100871, China (\email{songyujian@stu.pku.edu.cn})}
\and Aihui Zhou\footnotemark[2]}

\usepackage{amsopn}

\newcommand{\inner}[1]{\left\langle #1 \right\rangle}
\usepackage[version=4]{mhchem}
\usepackage{algorithm}
\usepackage{algpseudocode}
\usepackage{esint}
\usepackage{dsfont}
\makeatletter
\newcommand*\l@section[2]{
  \ifnum \c@tocdepth >\z@
    \addpenalty\@secpenalty
    \addvspace{1.0em \@plus\p@}
    \setlength\@tempdima{1.5em}
    \begingroup
      \parindent \z@ \rightskip \@pnumwidth
      \parfillskip -\@pnumwidth
      \leavevmode \bfseries
      \advance\leftskip\@tempdima
      \hskip -\leftskip
      #1\nobreak
      \leaders\hbox{$\m@th\mkern \@dotsep mu\hbox{.}\mkern \@dotsep mu$}\hfill
      \nobreak\hb@xt@\@pnumwidth{\hss #2}\par
    \endgroup
  \fi}

\newcommand*\l@subsection[2]{
  \ifnum \c@tocdepth >\z@
    \addpenalty\@secpenalty
    \addvspace{0.5em \@plus\p@}
    \setlength\@tempdima{2.3em}
    \begingroup
      \parindent \z@ \rightskip \@pnumwidth
      \parfillskip -\@pnumwidth
      \leavevmode \bfseries
      \advance\leftskip\@tempdima
      \hskip -\leftskip
      #1\nobreak
      \leaders\hbox{$\m@th\mkern \@dotsep mu\hbox{.}\mkern \@dotsep mu$}\hfill
      \nobreak\hb@xt@\@pnumwidth{\hss #2}\par
    \endgroup
  \fi}

\renewcommand{\tableofcontents}{
    \@starttoc{toc}
}
\makeatother

\begin{document}
	
	\maketitle
	
	\begin{abstract}
Many novel and unique physical phenomena in incommensurate systems can be illustrated and predicted using their spectral structure and electronic state distributions. However, the absence of periodicity in these systems poses significant challenges for obtaining the associated information. In this paper, by embedding the system into higher dimensions together with introducing a regularization technique, we prove that the spectrum of the Schrödinger operator for the incommensurate system can be approximated by the spectra of a family of regularized Schrödinger operators, which are elliptic, retain periodicity, and enjoy favorable analytic and spectral properties. We also show the well-posedness of the probability density describing the electronic state distribution of the incommensurate system, which can be approximated by the ones generated by the Bloch solutions to the regularized model. Our analysis provides theoretical support for understanding and computing incommensurate systems.
\end{abstract}

	
	
	

		\section{Introduction}\label{sec1}
Incommensurate systems have drawn significant interest from physicists and materials scientists due to their unique electronic, optical, and mechanical properties \cite{andrei2020graphene, britnell2013strong, cao2018correlated, cao2018unconventional,geim2013van,   kennes2021moire,xu2013graphene}. The incommensurate structure is produced by stacking the single layers of low-dimensional materials on top of one
another with small misalignment. When two atomically thin layers are stacked with a slight twist in orientation, they form moiré patterns and fundamental electronic behaviors appear, such as flat electronic bands and van Hove singularities near the Fermi level (see, e.g., \cite{cao2018unconventional, carr2017twistronics, dean2013hofstadter, woods2014commensurate}). For instance, the flat bands amplify electronic correlations, leading to emergent phenomena such as superconductivity, correlated insulator states, and unconventional magnetism, as prominently demonstrated in twisted bilayer graphene at the magic angle \cite{cao2018correlated, dindorkar2023magical, hao2024robust}. It is significant to investigate these structures from both theoretical and computational perspectives, which involves the spectral analysis of the systems and helps us predict and control their properties by changing material type, twist angle, strain, or external fields. 

Simulating incommensurate systems is computationally challenging primarily because the lack of a common periodicity destroys the translational symmetry that many numerical methods exploit. Without Bloch’s theorem, one cannot reduce the problem to a single unit cell, rendering many efficient algorithms ineffective and causing the required computational cost to increase dramatically. One conventional method for simulating incommensurate systems is to construct a supercell approximation with artificial strain \cite{koda2016coincidence, komsa2013electronic, loh2015graphene}.  Therefore, the techniques commonly employed for periodic systems, such as those based on Bloch’s theorem, can be applied. However, to achieve reasonable computational accuracy, the supercell sizes required are often extremely large, resulting in a substantial computational cost. Moreover, to the best of our knowledge, there is no mathematical support for the supercell approximation approach in the literature. The tight-binding models are also applied to the computation of incommensurate systems \cite{cances2023simple, etter2020modeling, massatt2018incommensurate, massatt2017electronic, watson2023bistritzer}, in which there are inherent modeling errors. The planewave approximation framework proposed in \cite{ZHOU201999} and developed in \cite{chen2021plane, dai2023extended, CiCP-30-5, doi:10.1137/23M1553650, CiCP-35-4} is both model-accurate and computationally efficient. We note that \cite{doi:10.1137/23M1553650} provides the convergence analysis of the planewave approximation for the density of states defined in the sense of the thermodynamic limit. However, rigorous mathematical analysis of the spectral properties and electronic state distributions in incommensurate systems remains incomplete, posing a significant challenge to a deeper understanding and reliable numerical computation of these systems.

We see that the incommensurate system is a typical quasi-periodic system. The study of the Schrödinger operator with the quasi-periodic potential has attracted considerable interest among mathematicians. The spectra of one-dimensional quasi-periodic systems were investigated in e.g., \cite{avron1981almost, avron1983almost,dinaburg1975one, eliasson1992floquet}. In \cite{karpeshina2019extended}, the existence of an absolutely continuous spectrum at high energy in the two-dimensional case was discussed. Nevertheless, it is significant and open to understand the properties of the spectrum of the quasi-periodic Schrödinger operators in multi-dimensions.

In this paper, we aim to investigate the spectrum of the following Schrödinger operator for the incommensurate system
\begin{align*}
    \mathscr{H}:=-\frac{1}{2}\Delta+V_{1}(\mathbf{r})+V_{2}(\mathbf{r})
\end{align*}
defined in $L^{2}(\mathbb{R}^{d})$ on the domain $H^{2}(\mathbb{R}^{d}) (d=1,2)$, where each potential $V_{j} (j=1,2)$ is periodic, while the two periods are incommensurate  (see more details in Section 2). We prove that the spectrum of $\mathscr{H}$ can be
approximated by the spectra of a family of regularized Schrödinger operators, which are elliptic,
retain periodicity, and enjoy favorable analytic and spectral properties. We also show the well-posedness of the probability density, and the probability density can be approximated by the ones generated by the Bloch solutions to the regularized model.

We see that it is quite convenient to restore the periodicity, for which we embed and view the system into  higher dimensions and investigate the following extended Schrödinger operator \cite{ZHOU201999}
\begin{align*}
    \tilde{\mathscr{H}}:=-\frac{1}{2}\sum_{i=1}^{d}(\partial_{\mathbf{r}_{i}}+\partial_{\mathbf{r}_{i}'})^{2}+V_{1}(\mathbf{r})+V_{2}(\mathbf{r}')
\end{align*}
in $L^{2}(\mathbb{R}^{2d})$. However, we observe that the operator $\tilde{\mathscr{H}}$ is
degenerate elliptic and is not self-adjoint on the domain $H^{2}(\mathbb{R}^{2d})$. To avoid addressing a non-self-adjoint operator, we perform a self-adjointization by constructing a self-adjoint extension of $\tilde{\mathscr{H}}$ and show the spectral consistency between this extension and the original incommensurate operator $\mathscr{H}$ (see Theorem \ref{prop:susbs}). We also apply the Bloch-Floquet transform to this extension and obtain the corresponding direct integral decomposition (see details in Section 3). 
To handle the degenerate ellipticity of $\tilde{\mathscr{H}}$, in particular, we introduce a regularization of the extended Schrödinger operator (see details in Section 4): given $\delta>0$, 
\begin{align*}
    \tilde{\mathscr{H}}^{\delta}:= \tilde{\mathscr{H}}-\frac{\delta}{2}\sum_{i=1}^{d}(\partial_{\mathbf{r}_{i}}-\partial_{\mathbf{r}_{i}'})^{2}
 \end{align*}
in $L^{2}(\mathbb{R}^{2d})$ on the domain $H^{2}(\mathbb{R}^{2d})$. By the approximation theory of the spectrum and the spectral consistency between $\mathscr{H}$ and $\tilde{\mathscr{H}}$, we prove that the spectrum of the Schrödinger operator $\mathscr{H}$ for the incommensurate system can be approximated by that of the regularized operator $\tilde{\mathscr{H}}^{\delta}$ as $\delta\rightarrow0^{+}$. 
In addition, we establish the well-posedness of the probability density associated with electronic states in the incommensurate system and show that the probability density can be approximated by those generated from Bloch solutions of the regularized model. Such a regularized model is well-posed: all relevant physical observable quantities are rigorously defined, and the classical Bloch theory, together with a broad range of efficient numerical algorithms developed for periodic systems, is applicable. Our analysis provides theoretical support and an efficient approach for understanding and computing the physical observable quantities for incommensurate systems, and can be extended to systems in any dimension and with any number of stacked material layers.

The rest of this paper is organized as follows. We introduce the incommensurate system, its embedding into higher dimensions, and the classical Bloch-Floquet transform together with relevant notation and existing results in Section 2. In Section 3, we construct a self-adjoint extension of the extended Schrödinger operator and apply the Bloch-Floquet transform to the extension. In Section 4, we introduce a regularization technique and investigate the relationship between the spectra of the extended Schrödinger operator and the regularized Schrödinger operator. In Section 5, by establishing the spectral consistency between the original and extended Schrödinger operators, we derive the property of the spectrum of the Schrödinger operator for the incommensurate system, together with the property of the probability density that describes the electronic state distributions. Finally, we give some concluding remarks in Section 6.

\section{Preliminaries}
In this section, we introduce the notation for the Schrödinger equation for the incommensurate system and its embedding into higher dimensions, together with the analytical tools and results that will be used in the subsequent analysis. 

\subsection{Incommensurate system} 
 Consider two periodic systems in $d$ dimensions stacked in parallel along the $(d+1)$th dimension ($d=1, 2$). For simplicity, we neglect the $(d+1)$th dimension and the distance between the layers.

Each $d$-dimensional periodic system can be described by a Bravais lattice
\begin{align*}
    \mathcal{R}_{j}:=\{A_{j}n:n\in\mathbb{Z}^{d}\},\quad j=1,2,
\end{align*}
where $A_{j}(j=1,2)\in\mathbb{R}^{d\times d}$ is invertible. We denote the unit cell for the $j$-th layer by 
\begin{align*}
    \Gamma_{j}:=\{A_{j}\alpha:\alpha\in[0,1)^{d}\},\quad j=1,2.
\end{align*}
Subsequently, the associated reciprocal lattice and reciprocal unit cell are given by
\begin{align*}
    \mathcal{R}_{j}^{*}:=\{2\pi A_{j}^{-\mathrm{T}}n:n\in\mathbb{Z}^{d}\},\quad \Gamma_{j}^{*}:=\{2\pi A_{j}^{-\mathrm{T}}\alpha:\alpha\in[0,1)^{d}\},\quad j=1,2.
\end{align*}

We note that each $\mathcal{R}_{j}$ is periodic with the translation invariance with respect to its lattice vectors:
\begin{align*}
    \mathcal{R}_{j}=A_{j}n+\mathcal{R}_{j},\quad\forall n\in\mathbb{Z}^{d},j=1,2.
\end{align*}
However, the stacked system $\mathcal{R}_{1}\cup\mathcal{R}_{2}$ may lose the periodicity and the translation invariance property. It follows the definition of the incommensurate system consisting of two periodic lattices.
\begin{definition}
Two Bravais lattices $\mathcal{R}_{1}$ and $\mathcal{R}_{2}$ are called incommensurate if
\begin{align*}
    \mathcal{R}_{1}\cup\mathcal{R}_{2}+\tau=\mathcal{R}_{1}\cup\mathcal{R}_{2}\Longleftrightarrow\tau=\mathbf{0}\in\mathbb{R}^{d}.
\end{align*}
\end{definition}

In this paper, we shall use the standard notation for Sobolev spaces with associated norms (see, e.g. \cite{adams2003sobolev}).
We focus on the following Schrödinger operator $\mathscr{H}$ in $L^{2}(\mathbb{R}^{d})$ on the domain $H^{2}(\mathbb{R}^{d})$ for the bilayer incommensurate system,
\begin{align*}
    \mathscr{H}:=-\frac{1}{2}\Delta+V_{1}(\mathbf{r})+V_{2}(\mathbf{r}),
\end{align*}
where the potential $V_{j}: \mathbb{R}^{d}\rightarrow \mathbb{R}$ is $\mathcal{R}_{j}$-periodic, i.e., 
\begin{align*}
    V_{j}(\mathbf{r}+\tau)=V_{j}(\mathbf{r}),\quad\forall\mathbf{r}\in\mathbb{R}^{d},\quad\forall\tau\in\mathcal{R}_{j},\quad j=1,2,
\end{align*}
and $\mathcal{R}_{1}$ and $\mathcal{R}_{2}$ are incommensurate.
We assume that $V_{j}\in C^{0}(\Gamma_{j}) (j=1,2)$ in this paper.

 The absence of periodicity leads to fundamental difficulties for the study of the Schrödinger operator $\mathscr{H}$. Notably, the bilayer incommensurate system can be embedded in a higher-dimensional space $\mathbb{R}^{d}\times\mathbb{R}^{d}$ and the periodicity can be restored by the extended Schrödinger operator \cite{ZHOU201999}
\begin{align*}
    \tilde{\mathscr{H}}:=-\frac{1}{2}\tilde{D}+V_{1}(\mathbf{r})+V_{2}(\mathbf{r}')
\end{align*}
in $L^{2}(\mathbb{R}^{2d})$ on the domain $H^{2}(\mathbb{R}^{2d})$, where for $\tilde{\mathbf{r}}:=(\mathbf{r}, \mathbf{r}')\in\mathbb{R}^{d}\times\mathbb{R}^{d}$ and $\mathbf{r}=(\mathbf{r}_{1},\ldots,\mathbf{r}_{d})$, $\mathbf{r}'=(\mathbf{r}_{1}',\ldots,\mathbf{r}_{d}')$,  $\tilde{D}$ is defined by
\begin{align*}
    \left(\tilde{D}\tilde{u}\right)(\mathbf{r}, \mathbf{r}' ):=\left(\sum_{i=1}^{d}(\partial_{\mathbf{r}_{i}}+\partial_{\mathbf{r}_{i}'})^{2}\tilde{u}\right)(\mathbf{r}, \mathbf{r}' ).
\end{align*}

Since $V_{1}$ and $V_{2}$ are periodic in $\mathbb{R}^{d}$ with respect to $\mathcal{R}_{1}$ and $\mathcal{R}_{2}$, respectively, the potential $V_{1}(\mathbf{r})+V_{2}(\mathbf{r}')$ is periodic in $\mathbb{R}^{d}\times\mathbb{R}^{d}$ with respect to a higher dimensional Bravais lattice
\begin{align*}
    \tilde{\mathcal{R}}:=\mathcal{R}_{1}\times\mathcal{R}_{2}=\left\{(A_{1}m, A_{2}n):(m,n)\in\mathbb{Z}^{d}\times\mathbb{Z}^{d}\right\}.
\end{align*}
Therefore, the operator $\tilde{\mathscr{H}}$ is periodic and translation invariant with respect to $\tilde{\mathcal{R}}$. Subsequently, we introduce the unit cell for the lattice $\tilde{\mathcal{R}}$ by 
\begin{align*}
    \tilde{\Gamma}:=\Gamma_{1}\times\Gamma_{2},
\end{align*}
and the associated reciprocal lattice and reciprocal unit cell are given by
\begin{align*}
    \tilde{\mathcal{R}}^{*}:=\mathcal{R}_{1}^{*}\times\mathcal{R}_{2}^{*},\quad
    \tilde{\Gamma}^{*}:=\Gamma_{1}^{*}\times\Gamma_{2}^{*}.
\end{align*}

The Schrödinger equation 
\begin{align}\label{eq:incom_sys}
    \mathscr{H}u=\lambda u
\end{align}
for the incommensurate system is then transformed into the extended Schrödinger equation \begin{align}\label{eq:perio_sys}
    \tilde{\mathscr{H}}\tilde{u}=\lambda \tilde{u}
\end{align} for the periodic system. In Section 5, we will show the spectral consistency between  $\mathscr{H}$ and $\tilde{\mathscr{H}}$ under some assumptions.

\subsection{Bloch-Floquet transform}
Note that the periodicity of the extended Schrödinger operator $\tilde{\mathscr{H}}$ with respect to the Bravais lattice $\tilde{\mathcal{R}}$ is restored. To study the property of $\tilde{\mathscr{H}}$, the following Bloch-Floquet transform on $L^{2}(\mathbb{R}^{2d})$ will be used (see, e.g., \cite{gontier2016convergence, kuchment2016overview, li2024mathematical}).

\begin{definition}\label{def:b_f_trans}
The Bloch-Floquet transform $\mathcal{U}_{\tilde{\mathcal{R}}}$ on $ L^{2}(\mathbb{R}^{2d})$ acts as follows:
\begin{equation*}
    f(\tilde{\mathbf{r}})\mapsto\mathcal{U}_{\tilde{\mathcal{R}}}f(\tilde{\mathbf{r}},\tilde{\mathbf{k}}):=\sum_{\tilde{\tau}\in\tilde{\mathcal{R}}}f(\tilde{\mathbf{r}}+\tilde{\tau})e^{-\mathrm{i}\tilde{\mathbf{k}}\cdot \tilde{\tau}},\quad \forall f\in L^{2}(\mathbb{R}^{2d}).
\end{equation*}
\end{definition}

Due to $\mathbb{R}^{2d}=\bigcup_{\tilde{\tau}\in \tilde{\mathcal{R}}}(\tilde{\Gamma}+\tilde{\tau})$, we observe that $\mathcal{U}_{\tilde{\mathcal{R}}}f(\tilde{\mathbf{r}},\tilde{\mathbf{k}})$ indeed amounts to taking the Fourier series of $f(\tilde{\mathbf{r}})$ over each shifted unit cell $\tilde{\Gamma}+\tilde{\tau}$. It is easy to see that $\mathcal{U}_{\tilde{\mathcal{R}}}f(\tilde{\mathbf{r}},\tilde{\mathbf{k}})$ adjusts the amplitude periodically involving $\tilde{\mathbf{r}}$ over $\tilde{\mathcal{R}}$ and is periodic involving $\tilde{\mathbf{k}}$ over $\tilde{\mathcal{R}}^{*}$. That is,
\begin{align*}
     \mathcal{U}_{\tilde{\mathcal{R}}}f(\tilde{\mathbf{r}}+\tilde{\tau}, \tilde{\mathbf{k}})=&e^{\mathrm{i}\tilde{\mathbf{k}}\cdot \tilde{\tau}}\mathcal{U}_{\tilde{\mathcal{R}}}f(\tilde{\mathbf{r}},\tilde{\mathbf{k}}),\quad \tilde{\tau}\in\tilde{\mathcal{R}},\\
    \mathcal{U}_{\tilde{\mathcal{R}}}f(\tilde{\mathbf{r}},\tilde{\mathbf{k}}+\tilde{\tau}^{*})=&\mathcal{U}_{\tilde{\mathcal{R}}}f(\tilde{\mathbf{r}},\tilde{\mathbf{k}}), \quad \tilde{\tau}^{*}\in\mathcal{R}^{*}.
\end{align*}
And the inverse of $\mathcal{U}_{\tilde{\mathcal{R}}}$ is given by averaging over the Brillouin zone,  
\begin{align}\label{eq:inverse}
    f(\tilde{\mathbf{r}})=\fint_{\tilde{\Gamma}^{*}}\mathcal{U}_{\tilde{\mathcal{R}}}f(\tilde{\mathbf{r}},\tilde{\mathbf{k}})d\tilde{\mathbf{k}}, \quad \tilde{\mathbf{r}}\in\mathbb{R}^{2d},
\end{align}
where 
\begin{align*}
   \fint_{\tilde{\Gamma}^{*}}:= \frac{1}{|\tilde{\Gamma}^{*}|}\int_{\tilde{\Gamma}^{*}}. 
\end{align*}

There hold the following conclusions for the Bloch-Floquet transform $\mathcal{U}_{\tilde{\mathcal{R}}}$ acting on $L^{2}(\mathbb{R}^{2d})$ (see, e.g., \cite{gontier2016convergence, kuchment2012floquet,kuchment2016overview}).
\begin{lemma}\label{lemma:isome}
The Bloch-Floquet transform $\mathcal{U}_{\tilde{\mathcal{R}}}$ acting on $L^{2}(\mathbb{R}^{2d})$ has the following properties:
    \begin{enumerate}[label=(\alph*)]
        \item  $\mathcal{U}_{\tilde{\mathcal{R}}}$ is isometry from $L^{2}(\mathbb{R}^{2d})$ onto $L^{2}(\tilde{\Gamma}^{*}, L^{2}(\tilde{\Gamma}))$;
        \item $\mathcal{U}_{\tilde{\mathcal{R}}}$ is isometry from $H^{2}(\mathbb{R}^{2d})$ onto $L^{2}(\tilde{\Gamma}^{*}, \mathcal{H})$, where
        \begin{align*}
        \mathcal{H}:=&\bigcup_{\tilde{\mathbf{k}}\in\tilde{\Gamma}^{*}}H_{\tilde{\mathbf{k}}}^{2}(\tilde{\Gamma}),\\
           H_{\tilde{\mathbf{k}}}^{2}(\tilde{\Gamma}):=&\{f|_{\tilde{\Gamma}}: f\in H_{\mathrm{loc}}^{2}(\mathbb{R}^{2d}), f(\tilde{\mathbf{r}}+\tilde{\tau})=e^{\mathrm{i}\tilde{\mathbf{k}}\cdot \tilde{\tau}}f(\tilde{\mathbf{r}})\,  a.e., \forall \tilde{\tau}\in\tilde{\mathcal{R}}\}.
        \end{align*}
    \end{enumerate}
\end{lemma}

Combining Lemma \ref{lemma:isome} and the concept of direct integral decompositions (see, e.g., \cite{corwin1990representations, nielsen2020direct}), we have the following conclusion.
\begin{lemma}\label{lem:direct_integral}
For the operator $\tilde{\mathscr{H}}$ in $L^{2}(\mathbb{R}^{2d})$ on the domain $H^{2}(\mathbb{R}^{2d})$, there hold the following direct integral decompositions 
\begin{align*}
         L^{2}(\mathbb{R}^{2d})=\fint_{\tilde{\Gamma}^{*}}^{\oplus}L^{2}(\tilde{\Gamma}),\quad
         H^{2}(\mathbb{R}^{2d})=\fint_{\tilde{\Gamma}^{*}}^{\oplus}H_{ \tilde{\mathbf{k}}}^{2}(\tilde{\Gamma}),\quad  \tilde{\mathscr{H}}=\fint_{\tilde{\Gamma}^{*}}^{\oplus}\tilde{\mathscr{H}}|_{H_{ \tilde{\mathbf{k}}}^{2}(\tilde{\Gamma})}.
    \end{align*}
\end{lemma}

For $\tilde{\mathbf{k}}\in\tilde{\Gamma}^{*}$ and $f\in H^{2}(\mathbb{R}^{2d})$, we obtain from Lemma \ref{lem:direct_integral} that 
\begin{equation*}
    \tilde{\mathscr{H}}(\mathcal{U}_{\tilde{\mathcal{R}}}f)(\cdot, \tilde{\mathbf{k}})=\tilde{\mathscr{H}}|_{H_{ \tilde{\mathbf{k}}}^{2}(\tilde{\Gamma})}(\mathcal{U}_{\tilde{\mathcal{R}}}f(\cdot,\tilde{\mathbf{k}})),
\end{equation*}
 which implies that $\tilde{\mathscr{H}}$ is block diagonal involving $\tilde{\mathbf{k}}\in\tilde{\Gamma}^{*}$ in the view of the Bloch-Floquet transform. Consequently, we focus on the property of  $\tilde{\mathscr{H}}$ restricted on each domain $H_{\tilde{\mathbf{k}}}^{2}(\tilde{\Gamma})$ for $\tilde{\mathbf{k}}\in\tilde{\Gamma}^{*}$.
 
 A direct calculation shows that
\begin{align}\label{eq:discr}
	H_{\tilde{\mathbf{k}}}^{2}(\tilde{\Gamma})=e^{\mathrm{i}\tilde{\mathbf{k}}\cdot\tilde{\mathbf{r}}}\left(H^{2}(\mathbb{T}^{2d})\right),\quad\forall \tilde{\mathbf{k}}\in\tilde{\Gamma}^{*},
\end{align}
where the torus $\mathbb{T}^{2d}:=\mathbb{R}^{2d}/\tilde{\mathcal{R}}$ and
\begin{align*}
 L^{2}(\mathbb{T}^{2d}):=&\{v|_{\tilde{\Gamma}}: v\in L_{\mathrm{loc}}^{2}(\mathbb{R}^{2d}),   v(\tilde{\mathbf{r}}+\tilde{\tau})=v(\tilde{\mathbf{r}})\,  a.e., \forall \tilde{\tau}\in\tilde{\mathcal{R}}\},\\
    H^{2}(\mathbb{T}^{2d}):=&\{v|_{\tilde{\Gamma}}: v\in H^{2}(\tilde{\Gamma}),  v(\tilde{\mathbf{r}}+\tilde{\tau})=v(\tilde{\mathbf{r}})\,  a.e., \forall \tilde{\tau}\in\tilde{\mathcal{R}}\}.
\end{align*}
We see that
\begin{align*}
    \tilde{\mathscr{H}}(\tilde{\mathbf{k}}):=e^{-\mathrm{i}\tilde{\mathbf{k}}\cdot\tilde{\mathbf{r}}}\circ\tilde{\mathscr{H}}\circ e^{\mathrm{i}\tilde{\mathbf{k}}\cdot\tilde{\mathbf{r}}}=\tilde{\mathscr{H}}-\mathrm{i}(\mathbf{k}+\mathbf{k}')\cdot(\nabla_{\mathbf{r}}+\nabla_{\mathbf{r}'})+\frac{1}{2}|\mathbf{k}+\mathbf{k}'|^{2}
\end{align*}
on the domain $H^{2}(\mathbb{T}^{2d})$, where $\mathbf{k}\in\Gamma_{1}^{*}, \mathbf{k}'\in\Gamma_{2}^{*}$ and $\tilde{\mathbf{k}}=(\mathbf{k}, \mathbf{k}')\in\tilde{\Gamma}^{*}$. Therefore, we have two equivalent representations for the block diagonal part of $\tilde{\mathscr{H}}$: one is a fixed operator $\tilde{\mathscr{H}}$ acting on a family of domains $\{H_{ \tilde{\mathbf{k}}}^{2}(\tilde{\Gamma})\}_{\tilde{\mathbf{k}}\in\tilde{\Gamma}^{*}}$; the other is a family of differential operators $\{\tilde{\mathscr{H}}(\tilde{\mathbf{k}})\}_{\tilde{\mathbf{k}}\in \tilde{\Gamma}^{*}}$ with one same domain $H^{2}(\mathbb{T}^{2d})$. Both of the representations will be used in the subsequent analysis.
 
\section{The extended Schrödinger operator $\tilde{\mathscr{H}}$}
To study the spectrum of the quasi-periodic Schrödinger operator $\mathscr{H}$, we first investigate the spectrum of the extended Schrödinger operator $\tilde{\mathscr{H}}$ embedded in higher dimensions. In Section 5, we will show the relationship between the spectrum of $\mathscr{H}$ and the spectrum of $\tilde{\mathscr{H}}$ after a self-adjoint extension.

\subsection{Self-adjoint extension}
We observe that $\tilde{\mathscr{H}}$ is not elliptic and its principal symbol
\begin{align*}
    \frac{1}{2}\sum_{i=1}^{d}(\xi_{i}+\xi_{i}')^{2}
\end{align*}
vanishes along the direction $\xi+\xi'=0$, where $\tilde{\xi}:=(\xi, \xi')$ and $\xi=(\xi_{1},\ldots,\xi_{d}), \xi'=(\xi_{1}',\ldots,\xi_{d}')\in\mathbb{R}^{d}$. That is, the operator $\tilde{\mathscr{H}}$ is degenerate along $\xi+\xi'=0$ in the frequency domain. Then there does not hold the Gårding's inequality and the `graph norm' $\Vert\tilde{\mathscr{H}}\cdot\Vert_{L^{2}}+\Vert\cdot\Vert_{L^{2}}$ is not equivalent to the norm $\Vert\cdot\Vert_{H^{2}}$. Therefore,  $\tilde{\mathscr{H}}$ is not even a closed operator on the domain $H^{2}(\mathbb{R}^{2d})$ and the spectrum $\sigma(\tilde{\mathscr{H}}|_{H^{2}(\mathbb{R}^{2d})})=\mathbb{C}$ (see Proposition 1.5.6 in \cite{de2008intermediate}). Hence, $\tilde{\mathscr{H}}$ on the domain $H^{2}(\mathbb{R}^{2d})$ is not self-adjoint.
Fortunately, we have
\begin{proposition}\label{prop:ess_a}
The extended Schrödinger operator $\tilde{\mathscr{H}}: L^{2}(\mathbb{R}^{2d})\rightarrow L^{2}(\mathbb{R}^{2d})$ on the domain $H^{2}(\mathbb{R}^{2d})$ is essentially self-adjoint.
\end{proposition}
\begin{proof}
It is clear that $V_{1}(\mathbf{r})+V_{2}(\mathbf{r}')\in  L^{\infty}(\tilde{\Gamma})$ and as an operator
\begin{align*}
    L^{2}(\mathbb{R}^{2d})\rightarrow L^{2}(\mathbb{R}^{2d})
\end{align*}
on the domain $H^{2}(\mathbb{R}^{2d})$ is self-adjoint. Next, we consider the symmetric operator $\tilde{D}=\sum_{i=1}^{d}(\partial_{\mathbf{r}_{i}}+\partial_{\mathbf{r}_{i}'})^{2}$. For $u\in \operatorname{ker}(\tilde{D}^{*}\pm\mathrm{i}I)$, there holds that
\begin{align*}
    \mp\mathrm{i} \mathcal{F}_{\mathbb{R}^{2d}}u(\tilde{\xi})= \mathcal{F}_{\mathbb{R}^{2d}}(\tilde{D}^{*}u)(\tilde{\xi})=\sum_{i=1}^{d}(\xi_{i}+\xi_{i}')^{2} \mathcal{F}_{\mathbb{R}^{2d}}u(\tilde{\xi}),
\end{align*}
where $\mathcal{F}_{\mathbb{R}^{2d}}$ denotes the Fourier transform on $L^{2}(\mathbb{R}^{2d})$ and $I$ is the identity operator. Note that $\sum_{i=1}^{d}(\xi_{i}+\xi_{i}')^{2}\in\mathbb{R}$, then $\mathcal{F}_{\mathbb{R}^{2d}}u=0$ and $\operatorname{ker}(\tilde{D}^{*}\pm\mathrm{i}I)=\{0\}$, which implies that $\tilde{D}$ is essentially self-adjoint (see \cite{reed1980methods}). Therefore, by the Kato–Rellich theorem (see, e.g., \cite{kato2013perturbation, reed1980methods}), the symmetric operator 
\begin{align*}
    \tilde{\mathscr{H}}=-\frac{1}{2}\tilde{D}+V_{1}(\mathbf{r})+V_{2}(\mathbf{r}')
\end{align*}
on the domain $H^{2}(\mathbb{R}^{2d})$ is essentially self-adjoint.
\end{proof}

The essential self-adjointness of $\tilde{\mathscr{H}}$ on the domain $H^{2}(\mathbb{R}^{2d})$ implies that $\tilde{\mathscr{H}}$ has one and only one self-adjoint extension, i.e., the closure of $\tilde{\mathscr{H}}$. For simplicity, we still denote this closure by $\tilde{\mathscr{H}}: L^{2}(\mathbb{R}^{2d})\rightarrow L^{2}(\mathbb{R}^{2d})$ with the domain $\mathcal{G}$, where $H^{2}(\mathbb{R}^{2d})\subset\mathcal{G}\subset L^{2}(\mathbb{R}^{2d})$ is defined by
\begin{align*}
    \mathcal{G}:=\left\{u\in L^{2}(\mathbb{R}^{2d}): \Vert u\Vert^{2}_{L^{2}(\mathbb{R}^{2d})}+\Vert\tilde{\mathscr{H}} u\Vert^{2}_{L^{2}(\mathbb{R}^{2d})}<\infty\right\},
\end{align*}
and is equipped with the graph norm 
\begin{align*}
    \Vert\cdot\Vert_{\mathcal{G}}:=\left(\Vert\cdot\Vert^{2}_{L^{2}(\mathbb{R}^{2d})}+\Vert\tilde{\mathscr{H}}\cdot\Vert^{2}_{L^{2}(\mathbb{R}^{2d})}\right)^{1/2}.
\end{align*}

\subsection{Bloch-Floquet transform on $\mathcal{G}$}
The following conclusion about the operator $\tilde{\mathscr{H}}|_{H_{ \tilde{\mathbf{k}}}^{2}(\tilde{\Gamma})}$ follows from the similar discussions in Proposition \ref{prop:ess_a}. The only difference is that we need to make a slight adjustment to the quasi-periodic boundary conditions.
\begin{lemma}
For $\tilde{\mathbf{k}}\in\tilde{\Gamma}^{*}$, the operator $\tilde{\mathscr{H}}: L_{ \tilde{\mathbf{k}}}^{2}(\tilde{\Gamma})\rightarrow L_{ \tilde{\mathbf{k}}}^{2}(\tilde{\Gamma})$ on the domain $H_{\tilde{\mathbf{k}}}^{2}(\tilde{\Gamma})$ is essential self-adjoint, where 
\begin{align*}
    L_{ \tilde{\mathbf{k}}}^{2}(\tilde{\Gamma}):=\{f|_{\tilde{\Gamma}}: f\in L_{\mathrm{loc}}^{2}(\mathbb{R}^{2d}), f(\tilde{\mathbf{r}}+\tilde{\tau})=e^{\mathrm{i}\tilde{\mathbf{k}}\cdot \tilde{\tau}}f(\tilde{\mathbf{r}})\,  a.e., \forall \tilde{\tau}\in\tilde{\mathcal{R}}\}.
\end{align*}
\end{lemma}

Therefore, by the essential self-adjointness, the closure of $\tilde{\mathscr{H}}|_{H_{ \tilde{\mathbf{k}}}^{2}(\tilde{\Gamma})}$ is self-adjoint, which is denoted by $\tilde{\mathscr{H}}: L_{ \tilde{\mathbf{k}}}^{2}(\tilde{\Gamma})\rightarrow L_{ \tilde{\mathbf{k}}}^{2}(\tilde{\Gamma})$ with the domain $\{\mathcal{G}_{\tilde{\mathbf{k}}}(\tilde{\Gamma})\}_{\tilde{\mathbf{k}}\in\tilde{\Gamma}^{*}}$ for simplicity. For $\tilde{\mathbf{k}}\in\tilde{\Gamma}^{*}$, the domain $H_{\tilde{\mathbf{k}}}^{2}(\tilde{\Gamma})\subset\mathcal{G}_{\tilde{\mathbf{k}}}(\tilde{\Gamma})\subset  L_{ \tilde{\mathbf{k}}}^{2}(\tilde{\Gamma})$ is defined by
\begin{align*}
    \mathcal{G}_{\tilde{\mathbf{k}}}(\tilde{\Gamma}):=\left\{u\in  L_{ \tilde{\mathbf{k}}}^{2}(\tilde{\Gamma}): \Vert u\Vert^{2}_{L^{2}(\tilde{\Gamma})}+\Vert\tilde{\mathscr{H}} u\Vert^{2}_{L^{2}(\tilde{\Gamma})}<\infty\right\}
\end{align*}
and is equipped with the graph norm 
\begin{align*}
    \Vert\cdot\Vert_{\mathcal{G}(\tilde{\Gamma})}:=\left(\Vert\cdot\Vert^{2}_{L^{2}(\tilde{\Gamma})}+\Vert\tilde{\mathscr{H}}\cdot\Vert^{2}_{L^{2}(\tilde{\Gamma})}\right)^{1/2}.
\end{align*}

Recalling the Bloch-Floquet transform and the direct integral decomposition of $\tilde{\mathscr{H}}$  on the domain $H^{2}(\mathbb{R}^{2d})$ addressed in Lemmas \ref{lemma:isome} and \ref{lem:direct_integral}, we have following proposition.
\begin{proposition}\label{prop:b_int}
If the Bloch-Floquet transform $\mathcal{U}_{\tilde{\mathcal{R}}}$ is defined by Definition \ref{def:b_f_trans}, then $\mathcal{U}_{\tilde{\mathcal{R}}}$ is isometry from $\mathcal{G}$ onto $L^{2}(\tilde{\Gamma}^{*}, \bigcup_{\tilde{\mathbf{k}}\in\tilde{\Gamma}^{*}}\{\mathcal{G}_{\tilde{\mathbf{k}}}(\tilde{\Gamma})\}_{\tilde{\mathbf{k}}\in\tilde{\Gamma}^{*}})$. Therefore, 
\begin{align*}
    \mathcal{G}=\fint_{\tilde{\Gamma}^{*}}^{\oplus}\mathcal{G}_{\tilde{\mathbf{k}}}(\tilde{\Gamma}),\quad  \tilde{\mathscr{H}}|_{\mathcal{G}}=\fint_{\tilde{\Gamma}^{*}}^{\oplus}\tilde{\mathscr{H}}|_{\mathcal{G}_{\tilde{\mathbf{k}}}(\tilde{\Gamma})}.
\end{align*}
\end{proposition}
\begin{proof}
Note that for $\tilde{\mathbf{k}}\in\tilde{\Gamma}^{*}$ and $f\in \mathcal{G}$,
\begin{align*}
 \tilde{\mathscr{H}}\left(\mathcal{U}_{\tilde{\mathcal{R}}}f(\cdot, \tilde{\mathbf{k}})\right)=&\tilde{\mathscr{H}}\sum_{\tilde{\tau}\in\tilde{\mathcal{R}}}f(\cdot+\tilde{\tau})e^{-\mathrm{i}\tilde{\mathbf{k}}\cdot \tilde{\tau}}\\=&\sum_{\tilde{\tau}\in\tilde{\mathcal{R}}}\left(\tilde{\mathscr{H}}f\right)(\cdot+\tilde{\tau})e^{-\mathrm{i}\tilde{\mathbf{k}}\cdot \tilde{\tau}}=\left(\mathcal{U}_{\tilde{\mathcal{R}}}\left(\tilde{\mathscr{H}}f\right)\right)(\cdot, \tilde{\mathbf{k}}).
\end{align*}
We obtain from Lemma \ref{lemma:isome} and the definition of $\mathcal{G}$ that
\begin{equation}
    \begin{aligned}\label{eq:bb}
     &\int_{\tilde{\mathbf{k}}\in\tilde{\Gamma}^{*}}\left\Vert\left(\mathcal{U}_{\tilde{\mathcal{R}}}f\right)(\cdot, \tilde{\mathbf{k}})\right\Vert^{2}_{\mathcal{G}(\tilde{\Gamma})}d\tilde{\mathbf{k}}\\=&\int_{\tilde{\mathbf{k}}\in\tilde{\Gamma}^{*}}\left\Vert\left(\mathcal{U}_{\tilde{\mathcal{R}}}f\right)(\cdot, \tilde{\mathbf{k}})\right\Vert^{2}_{L^{2}(\tilde{\Gamma})}d\tilde{\mathbf{k}}+\int_{\tilde{\mathbf{k}}\in\tilde{\Gamma}^{*}}\left\Vert\tilde{\mathscr{H}}\left(\mathcal{U}_{\tilde{\mathcal{R}}}f(\cdot, \tilde{\mathbf{k}})\right)\right\Vert^{2}_{L^{2}(\tilde{\Gamma})}d\tilde{\mathbf{k}}
   \\=&\Vert f\Vert^{2}_{L^{2}(\mathbb{R}^{2d})}+\int_{\tilde{\mathbf{k}}\in\tilde{\Gamma}^{*}}\left\Vert\left(\mathcal{U}_{\tilde{\mathcal{R}}}\left(\tilde{\mathscr{H}}f\right)\right)(\cdot, \tilde{\mathbf{k}})\right\Vert^{2}_{L^{2}(\tilde{\Gamma})}d\tilde{\mathbf{k}}\\=&\Vert f\Vert^{2}_{L^{2}(\mathbb{R}^{2d})}+\Vert \tilde{\mathscr{H}}f\Vert^{2}_{L^{2}(\mathbb{R}^{2d})}=\Vert f\Vert^{2}_{\mathcal{G}}<\infty,
\end{aligned}
\end{equation}
which yields that 
\begin{align*}
    \mathcal{U}_{\tilde{\mathcal{R}}}f\in L^{2}(\tilde{\Gamma}^{*}, \bigcup_{\tilde{\mathbf{k}}\in\tilde{\Gamma}^{*}}\{\mathcal{G}_{\tilde{\mathbf{k}}}(\tilde{\Gamma})\}_{\tilde{\mathbf{k}}\in\tilde{\Gamma}^{*}}).
\end{align*}

 Since $\mathcal{G}\subset L^{2}(\mathbb{R}^{2d})$ and $L^{2}(\tilde{\Gamma}^{*}, \bigcup_{\tilde{\mathbf{k}}\in\tilde{\Gamma}^{*}}\{\mathcal{G}_{\tilde{\mathbf{k}}}(\tilde{\Gamma})\}_{\tilde{\mathbf{k}}\in\tilde{\Gamma}^{*}})\subset L^{2}(\tilde{\Gamma}^{*}, L^{2}(\tilde{\Gamma}))$, we see that $\mathcal{U}_{\tilde{\mathcal{R}}}|_{\mathcal{G}}$ is an injection by Lemma \ref{lemma:isome}. Thus, we only need to show $\mathcal{U}_{\tilde{\mathcal{R}}}|_{\mathcal{G}}$ is a surjection.

Since $L^{2}(\tilde{\Gamma}^{*}, \bigcup_{\tilde{\mathbf{k}}\in\tilde{\Gamma}^{*}}\{\mathcal{G}_{\tilde{\mathbf{k}}}(\tilde{\Gamma})\}_{\tilde{\mathbf{k}}\in\tilde{\Gamma}^{*}})\subset L^{2}(\tilde{\Gamma}^{*}, L^{2}(\tilde{\Gamma}))$ and $\mathcal{U}_{\tilde{\mathcal{R}}}$ is isometry from $L^{2}(\mathbb{R}^{2d})$ onto $L^{2}(\tilde{\Gamma}^{*}, L^{2}(\tilde{\Gamma}))$, for given $\mathcal{V}\in L^{2}(\tilde{\Gamma}^{*}, \bigcup_{\tilde{\mathbf{k}}\in\tilde{\Gamma}^{*}}\{\mathcal{G}_{\tilde{\mathbf{k}}}(\tilde{\Gamma})\}_{\tilde{\mathbf{k}}\in\tilde{\Gamma}^{*}})$, there exists $f\in L^{2}(\mathbb{R}^{2d})$ such that 
\begin{align*}
    \mathcal{V}=\mathcal{U}_{\tilde{\mathcal{R}}}f.
\end{align*}

Note that (\ref{eq:bb}) implies
\begin{align*}
    \Vert f\Vert_{\mathcal{G}}^{2}=\int_{\tilde{\mathbf{k}}\in\tilde{\Gamma}^{*}}\left\Vert\left(\mathcal{U}_{\tilde{\mathcal{R}}}f\right)(\cdot, \tilde{\mathbf{k}})\right\Vert^{2}_{\mathcal{G}(\tilde{\Gamma})}d\tilde{\mathbf{k}}=\int_{\tilde{\mathbf{k}}\in\tilde{\Gamma}^{*}}\left\Vert\mathcal{V}(\cdot, \tilde{\mathbf{k}})\right\Vert^{2}_{\mathcal{G}(\tilde{\Gamma})}d\tilde{\mathbf{k}}<\infty,
\end{align*}
which shows that $f\in \mathcal{G}$ and $\mathcal{U}_{\tilde{\mathcal{R}}}|_{\mathcal{G}}$ is a surjection. That is, $\mathcal{U}_{\tilde{\mathcal{R}}}$ is isometry from $\mathcal{G}$ onto $L^{2}(\tilde{\Gamma}^{*}, \bigcup_{\tilde{\mathbf{k}}\in\tilde{\Gamma}^{*}}\{\mathcal{G}_{\tilde{\mathbf{k}}}(\tilde{\Gamma})\}_{\tilde{\mathbf{k}}\in\tilde{\Gamma}^{*}})$.

This completes the proof.
\end{proof}

A direct calculation leads to
\begin{align*}
	\mathcal{G}_{\tilde{\mathbf{k}}}(\tilde{\Gamma})=e^{\mathrm{i}\tilde{\mathbf{k}}\cdot\tilde{\mathbf{r}}}\left(\mathcal{G}(\mathbb{T}^{2d})\right),\quad\forall \tilde{\mathbf{k}}\in\tilde{\Gamma}^{*},
\end{align*}
where 
\begin{align*}
    \mathcal{G}(\mathbb{T}^{2d}):=\{v|_{\tilde{\Gamma}}: v\in \mathcal{G}(\tilde{\Gamma}),  v(\tilde{\mathbf{r}}+\tilde{\tau})=v(\tilde{\mathbf{r}})\,  a.e., \forall \tilde{\tau}\in\tilde{\mathcal{R}}\},
\end{align*}
which together with Proposition \ref{prop:b_int} shows the closure of $\tilde{\mathscr{H}}(\tilde{\mathbf{k}})$ on the domain $H^{2}(\mathbb{T}^{2d})$. For simplicity, we denote this closure
by $\tilde{\mathscr{H}}(\tilde{\mathbf{k}})$ with the domain $\mathcal{G}(\mathbb{T}^{2d})$. We see that $\tilde{\mathscr{H}}(\tilde{\mathbf{k}})$ is self-adjoint on $\mathcal{G}(\mathbb{T}^{2d})$. Let $\sigma_{\mathcal{G}}(\tilde{\mathscr{H}}(\tilde{\mathbf{k}}))$ be the spectrum set of $\tilde{\mathscr{H}}(\tilde{\mathbf{k}})$ on the domain $\mathcal{G}(\mathbb{T}^{2d})$ for $\tilde{\mathbf{k}}\in\tilde{\Gamma}^{*}$ and there holds that
\begin{align*}
\sigma_{\mathcal{G}}(\tilde{\mathscr{H}}(\tilde{\mathbf{k}}))=\sigma(\tilde{\mathscr{H}}|_{\{\mathcal{G}_{\tilde{\mathbf{k}}}(\tilde{\Gamma})\}_{\tilde{\mathbf{k}}\in\tilde{\Gamma}^{*}}}),\quad \tilde{\mathbf{k}}\in\tilde{\Gamma}^{*}.
\end{align*}

Let $\sigma_{\mathcal{G}}(\tilde{\mathscr{H}})$ be the spectrum set of $\tilde{\mathscr{H}}$ on the domain $\mathcal{G}$. We have the relationship between $\sigma_{\mathcal{G}}(\tilde{\mathscr{H}})$ and $\sigma_{\mathcal{G}}(\tilde{\mathscr{H}}(\tilde{\mathbf{k}}))$ as follows.

\begin{theorem}\label{prop:dense}
There holds that
\begin{align*}
    \sigma_{\mathcal{G}}(\tilde{\mathscr{H}})\subset\overline{\bigcup_{\tilde{\mathbf{k}}\in\tilde{\Gamma}^{*}}\sigma_{\mathcal{G}}(\tilde{\mathscr{H}}(\tilde{\mathbf{k}}))}.
\end{align*}
That is, for $\tilde{\lambda}\in\sigma_{\mathcal{G}}(\tilde{\mathscr{H}})$ and $\varepsilon>0$, there exists $\tilde{\mathbf{k}}\in\tilde{\Gamma}^{*}$ and $\tilde{\lambda}_{\varepsilon}\in\sigma_{\mathcal{G}}(\tilde{\mathscr{H}}(\tilde{\mathbf{k}}))$ such that 
\begin{align*}
    \left|\tilde{\lambda}-\tilde{\lambda}_{\varepsilon}\right|<\varepsilon.
\end{align*}
\end{theorem}
\begin{proof}
Note that $\tilde{\mathscr{H}}$ on the domain $\mathcal{G}$ and $\tilde{\mathscr{H}}(\tilde{\mathbf{k}})$ on the domain $\mathcal{G}(\mathbb{T}^{2d})$ are self-adjoint. Hence, 
\begin{align*}
    \sigma_{\mathcal{G}}(\tilde{\mathscr{H}})\subset\mathbb{R}\quad\text{and}\quad\overline{\bigcup_{\tilde{\mathbf{k}}\in\tilde{\Gamma}^{*}}\sigma_{\mathcal{G}}(\tilde{\mathscr{H}}(\tilde{\mathbf{k}}))}\subset\mathbb{R}.
\end{align*}

If $\tilde{\lambda}\notin\overline{\bigcup_{\tilde{\mathbf{k}}\in\tilde{\Gamma}^{*}}\sigma_{\mathcal{G}}(\tilde{\mathscr{H}}(\tilde{\mathbf{k}}))}$, then there exists $\delta>0$ such that 
\begin{align*}
    (\tilde{\lambda}-\delta, \tilde{\lambda}+\delta)\cap\overline{\bigcup_{\tilde{\mathbf{k}}\in\tilde{\Gamma}^{*}}\sigma_{\mathcal{G}}(\tilde{\mathscr{H}}(\tilde{\mathbf{k}}))}=\emptyset.
\end{align*}
That is, 
\begin{align*}
    (\tilde{\lambda}-\delta, \tilde{\lambda}+\delta)\cap\sigma_{\mathcal{G}}(\tilde{\mathscr{H}}(\tilde{\mathbf{k}}))=\emptyset,\quad\forall \tilde{\mathbf{k}}\in\tilde{\Gamma}^{*},
\end{align*}
and the distance between $\tilde{\lambda}$ and $\sigma_{\mathcal{G}}(\tilde{\mathscr{H}}(\tilde{\mathbf{k}}))$ satisfies
\begin{align*}
    \operatorname{dist}(\tilde{\lambda}, \sigma_{\mathcal{G}}(\tilde{\mathscr{H}}(\tilde{\mathbf{k}}))):=\inf_{\hat{\lambda}\in\sigma_{\mathcal{G}}(\tilde{\mathscr{H}}(\tilde{\mathbf{k}}))}\left|\tilde{\lambda}-\hat{\lambda}\right|\geqslant\delta, \quad\forall \tilde{\mathbf{k}}\in\tilde{\Gamma}^{*}.
\end{align*}

Hence, for $\tilde{\mathbf{k}}\in\tilde{\Gamma}^{*}$, $\tilde{\lambda} I-\tilde{\mathscr{H}}(\tilde{\mathbf{k}})$ is invertible and the norm of $\left(\tilde{\lambda} I-\tilde{\mathscr{H}}(\tilde{\mathbf{k}})\right)^{-1}$ can be estimated as
\begin{align*}
    \left\Vert\left(\tilde{\lambda} I-\tilde{\mathscr{H}}(\tilde{\mathbf{k}})\right)^{-1} \right\Vert\leqslant\frac{1}{ \operatorname{dist}(\tilde{\lambda}, \sigma_{\mathcal{G}}(\tilde{\mathscr{H}}(\tilde{\mathbf{k}})))}\leqslant\frac{1}{\delta}, \quad \forall \tilde{\mathbf{k}}\in\tilde{\Gamma}^{*}.
\end{align*}

Define the following direct integral
decomposition by
\begin{align*}
    \mathscr{G}:=\fint_{\tilde{\Gamma}^{*}}^{\oplus}\left(\tilde{\lambda} I-\tilde{\mathscr{H}}(\tilde{\mathbf{k}})\right)^{-1}.
\end{align*}
We see that
\begin{align}
    \Vert\mathscr{G}\Vert=\operatorname{esssup}_{\tilde{\mathbf{k}}\in\tilde{\Gamma}^{*}} \left\Vert\left(\tilde{\lambda} I-\tilde{\mathscr{H}}(\tilde{\mathbf{k}})\right)^{-1} \right\Vert\leqslant\frac{1}{\delta}.
\end{align}
That is, $\mathscr{G}$ is bounded in $L^{2}(\mathbb{R}^{2d})$.

Next, we show that $\mathscr{G}$ is the inverse of $\tilde{\lambda} I-\tilde{\mathscr{H}}$. Indeed, we have
\begin{align*}
    \mathscr{G}\left(\tilde{\lambda} I-\tilde{\mathscr{H}}\right)=\fint_{\tilde{\Gamma}^{*}}^{\oplus}\left(\tilde{\lambda} I-\tilde{\mathscr{H}}(\tilde{\mathbf{k}})\right)^{-1}\left(\tilde{\lambda} I-\tilde{\mathscr{H}}(\tilde{\mathbf{k}})\right)=\fint_{\tilde{\Gamma}^{*}}^{\oplus}I=I.
\end{align*}
Similarly, $\left(\tilde{\lambda} I-\tilde{\mathscr{H}}\right)\mathscr{G}=I$. Consequently, combining with the boundedness of $\mathscr{G}$, we obtain that $\tilde{\lambda}\notin\sigma_{\mathcal{G}}(\tilde{\mathscr{H}})$, which completes the proof.
\end{proof}

\section{A regularized Schrödinger operator $\tilde{\mathscr{H}}^{\delta}$} In this section, to address the degenerate ellipticity of the extended Schrödinger operator $\tilde{\mathscr{H}}$,  we introduce a regularization technique and investigate the resulting regularized Schrödinger operator, including its fundamental properties and spectral results.
\subsection{Regularization} We regularize the extended Schrödinger operator $\tilde{\mathscr{H}}$  and introduce the following regularized Schrödinger operator.

\begin{definition}
 Given $\delta>0$, define the regularized Schrödinger operator $\tilde{\mathscr{H}}^{\delta}: L^{2}(\mathbb{R}^{2d})\rightarrow L^{2}(\mathbb{R}^{2d})$ on the domain $H^{2}(\mathbb{R}^{2d})$ by 
 \begin{align*}
    \tilde{\mathscr{H}}^{\delta}:= \tilde{\mathscr{H}}-\frac{\delta}{2}\sum_{i=1}^{d}(\partial_{\mathbf{r}_{i}}-\partial_{\mathbf{r}_{i}'})^{2}.
 \end{align*}
\end{definition}

We see that $\tilde{\mathscr{H}}^{\delta} (\delta>0)$ is uniformly elliptic due to its principal symbol 
\begin{align}\label{ine:ell}
\frac{1}{2}\sum_{i=1}^{d}(\xi_{i}+\xi_{i}')^{2}+\frac{\delta}{2}\sum_{i=1}^{d}(\xi_{i}-\xi_{i}')^{2}\geqslant\min\{1, \delta \}\sum_{i=1}^{d}\left(\xi_{i}^{2}+\xi_{i}'^{2}\right).
\end{align}

The following proposition follows directly from Lemmas \ref{lemma:isome} and \ref{lem:direct_integral}, which tells the direct integral decomposition for the operator $\tilde{\mathscr{H}}^{\delta}$ on the domain $H^{2}(\mathbb{R}^{2d})$.

\begin{proposition}\label{thm:direct_integral}
Given $\delta>0$, there holds
 \begin{align*}
    \tilde{\mathscr{H}}^{\delta}=\fint_{\tilde{\Gamma}^{*}}^{\oplus}\tilde{\mathscr{H}}^{\delta}|_{H_{ \tilde{\mathbf{k}}}^{2}(\tilde{\Gamma})}.
    \end{align*}   
\end{proposition}

Given $\delta>0$ and $\tilde{\mathbf{k}}\in \tilde{\Gamma}^{*}$, we obtain from (\ref{eq:discr}) and simple calculations that the following regularized operator
    \begin{align*}
    \tilde{\mathscr{H}}^{\delta}(\tilde{\mathbf{k}}):=&e^{-\mathrm{i}\tilde{\mathbf{k}}\cdot\tilde{\mathbf{r}}}\circ\tilde{\mathscr{H}}^{\delta}\circ e^{\mathrm{i}\tilde{\mathbf{k}}\cdot\tilde{\mathbf{r}}}\\=&\tilde{\mathscr{H}}(\tilde{\mathbf{k}})-\frac{\delta}{2}\sum_{i=1}^{d}(\partial_{\mathbf{r}_{i}}-\partial_{\mathbf{r}_{i}'})^{2}-\mathrm{i}\delta(\mathbf{k}-\mathbf{k}')\cdot(\nabla_{\mathbf{r}}-\nabla_{\mathbf{r}'})+\frac{\delta}{2}|\mathbf{k}-\mathbf{k}'|^{2}
\end{align*}
on the domain $H^{2}(\mathbb{T}^{2d})$, where $\mathbf{k}\in\Gamma_{1}^{*}, \mathbf{k}'\in\Gamma_{2}^{*}$ and $\tilde{\mathbf{k}}=(\mathbf{k}, \mathbf{k}')\in\tilde{\Gamma}^{*}$. The regularized operator $\tilde{\mathscr{H}}^{\delta}(\tilde{\mathbf{k}})$ provides another view of the operator $\tilde{\mathscr{H}}^{\delta}|_{H_{ \tilde{\mathbf{k}}}^{2}(\tilde{\Gamma})}$, and is self-adjoint on the domain $H^{2}(\mathbb{T}^{2d})$. We see that 
\begin{align*}
\sigma(\tilde{\mathscr{H}}^{\delta}|_{H_{ \tilde{\mathbf{k}}}^{2}(\tilde{\Gamma})})=\sigma(\tilde{\mathscr{H}}^{\delta}(\tilde{\mathbf{k}})),\quad \tilde{\mathbf{k}}\in\tilde{\Gamma}^{*}.
\end{align*} 

Next we address the spectral structure of the regularized Schrödinger operator $\tilde{\mathscr{H}}^{\delta}$ on the domain $H^{2}(\mathbb{R}^{2d})$.

 Given $\delta>0$, the operator $\tilde{\mathscr{H}}^{\delta}(\tilde{\mathbf{k}})$ on the domain $H^{2}(\mathbb{T}^{2d})$ is Fredholm for $\tilde{\mathbf{k}}\in\tilde{\Gamma}^{*}$ due to the uniform ellipticity of operators (see, e.g., \cite{hormander2007analysis}). Based on the Fredholm property and Proposition \ref{thm:direct_integral}, we have the following relationship between $\sigma(\tilde{\mathscr{H}}^{\delta})$ and $\sigma(\tilde{\mathscr{H}}^{\delta}(\tilde{\mathbf{k}}))$ from Theorem 4.5.1 in \cite{kuchment2012floquet}.

\begin{proposition}\label{prop:spec_dec}
For $\delta>0$, there holds that 
\begin{align*}
    \sigma(\tilde{\mathscr{H}}^{\delta})=\bigcup_{\tilde{\mathbf{k}}\in\tilde{\Gamma}^{*}}\sigma(\tilde{\mathscr{H}}^{\delta}|_{H_{ \tilde{\mathbf{k}}}^{2}(\tilde{\Gamma})})=\bigcup_{\tilde{\mathbf{k}}\in\tilde{\Gamma}^{*}}\sigma(\tilde{\mathscr{H}}^{\delta}(\tilde{\mathbf{k}})).
\end{align*}
\end{proposition}

The spectral structure of each regularized
operator $\tilde{\mathscr{H}}^{\delta}(\tilde{\mathbf{k}})$ on the domain $H^{2}(\mathbb{T}^{2d})$ for $\tilde{\mathbf{k}}\in\tilde{\Gamma}^{*}$ and $\delta>0$ is stated in the following proposition.
\begin{proposition}\label{prop:disrete_po}
For $\delta>0$ and $\tilde{\mathbf{k}}\in\tilde{\Gamma}^{*}$, the spectrum of $\tilde{\mathscr{H}}^{\delta}(\tilde{\mathbf{k}})$ on the domain $H^{2}(\mathbb{T}^{2d})$ consists of real discrete point spectrum. In particular, the eigenvalues of $\tilde{\mathscr{H}}^{\delta}(\tilde{\mathbf{k}})$ on the domain $H^{2}(\mathbb{T}^{2d})$ have ﬁnite multiplicity.
\end{proposition}
\begin{proof}
Note that the regularized operator $\tilde{\mathscr{H}}^{\delta}(\tilde{\mathbf{k}})$ is self-adjoint for $\tilde{\mathbf{k}}\in\tilde{\Gamma}^{*}$ and $\delta>0$, then its spectrum is real. If $I$ is the identical mapping and $\tilde{\lambda}^{\delta}\in\mathbb{C}$ is non-real, then $\left(\tilde{\lambda}^{\delta} I-\tilde{\mathscr{H}}^{\delta}(\tilde{\mathbf{k}})\right)^{-1}$ exists. We obtain from Rellich–Kondrachov theorem that $H^{2}(\mathbb{T}^{2d})$ is compactly embedded in $L^{2}(\mathbb{T}^{2d})$. Then we see that $\left(\tilde{\lambda}^{\delta} I-\tilde{\mathscr{H}}^{\delta}(\tilde{\mathbf{k}})\right)^{-1}$ maps $L^{2}(\mathbb{T}^{2d})$ into $H^{2}(\mathbb{T}^{2d})$ and is compact. Riesz-Schauder theory implies that the spectrum of $\left(\tilde{\lambda}^{\delta} I-\tilde{\mathscr{H}}^{\delta}(\tilde{\mathbf{k}})\right)^{-1}$ is discrete and consists of eigenvalues of ﬁnite multiplicity. Therefore, the spectrum of $\tilde{\mathscr{H}}^{\delta}(\tilde{\mathbf{k}})$ on the domain $H^{2}(\mathbb{T}^{2d})$ is real and also discrete,  and its eigenvalues have finite multiplicity. We complete the proof.
\end{proof}

 Proposition \ref{prop:disrete_po} tells that the spectrum of $\tilde{\mathscr{H}}^{\delta}(\tilde{\mathbf{k}})$ on the domain $H^{2}(\mathbb{T}^{2d})$ consists of real discrete point spectrum, i.e.,
\begin{equation*}
    \sigma(\tilde{\mathscr{H}}^{\delta}(\tilde{\mathbf{k}}))=\{\tilde{\lambda}^{\delta}_{j}(\tilde{\mathbf{k}}):\tilde{\lambda}^{\delta}_{1}(\tilde{\mathbf{k}})\leqslant \cdots \leqslant \tilde{\lambda}^{\delta}_{n}(\tilde{\mathbf{k}})\leqslant\cdots\}.
\end{equation*}
The standard perturbation theory (see, e.g., \cite{kato2013perturbation}) shows that each `band' function $\tilde{\lambda}^{\delta}_{j}(\tilde{\mathbf{k}})$ is continuous with respect to $\tilde{\mathbf{k}}\in \tilde{\Gamma}^{*}$, which together with Proposition \ref{prop:spec_dec} implies that the spectrum of the regularized operator $\tilde{\mathscr{H}}^{\delta}$ is absolutely continuous.

Consider the following regularized model: given $\delta>0$, for some $\tilde{\lambda}^{\delta}\in\mathbb{R}$, find the solution (generalized eigenfunction) $\tilde{u}^{\delta}$ satisfying
\begin{align}\label{eq:re_ehn}
    \tilde{\mathscr{H}}^{\delta}\tilde{u}^{\delta}=\tilde{\lambda}^{\delta}\tilde{u}^{\delta}.
\end{align}

The following proposition tells the relationship between the Bloch solutions to the equation (\ref{eq:re_ehn}) and the spectrum of the regularized elliptic operator $\tilde{\mathscr{H}}^{\delta}$ on the domain $H^{2}(\mathbb{R}^{2d})$ for $\delta>0$.
\begin{proposition}\label{thm:thm_bloch}
   Given $\delta>0$, $\tilde{\lambda}^{\delta}\in\sigma(\tilde{\mathscr{H}}^{\delta})$ is equivalent to the existence of $\tilde{\mathbf{k}}\in\tilde{\Gamma}^{*}$ and the Bloch solution $\tilde{u}^{\delta}\in H^{2}_{\mathrm{loc}}(\mathbb{R}^{2d})$ of (\ref{eq:re_ehn}) satisfying  
    \begin{equation}\label{eq:bloch}
        \tilde{u}^{\delta}(\tilde{\mathbf{r}}+\tilde{\tau})=e^{\mathrm{i}\tilde{\mathbf{k}}\cdot \tilde{\tau}}\tilde{u}^{\delta}(\tilde{\mathbf{r}}), \quad\forall \tilde{\tau}\in\tilde{\mathcal{R}}.
    \end{equation}
That is, there exists $\tilde{v}^{\delta}\in H^{2}(\mathbb{T}^{2d})$ such that 
\begin{align}\label{eq:blochhhhh}
    \tilde{u}^{\delta}(\tilde{\mathbf{r}})=e^{\mathrm{i}\tilde{\mathbf{k}}\cdot\tilde{\mathbf{r}}}\tilde{v}^{\delta}(\tilde{\mathbf{r}}).
\end{align}
\end{proposition}
\begin{proof} If $\tilde{\lambda}^{\delta}\in\sigma(\tilde{\mathscr{H}}^{\delta})$, then we obtain from Proposition \ref{prop:spec_dec} that there exists $\tilde{\mathbf{k}}\in\tilde{\Gamma}^{*}$ satisfying
\begin{align*}
    \tilde{\lambda}^{\delta}\in\sigma(\tilde{\mathscr{H}}^{\delta}(\tilde{\mathbf{k}}))=\sigma(\tilde{\mathscr{H}}^{\delta}|_{H_{\tilde{\mathbf{k}}}^{2}(\tilde{\Gamma})}).
\end{align*}
By Proposition \ref{prop:disrete_po}, we see that $\tilde{\lambda}^{\delta}$ is the eigenvalue of $\tilde{\mathscr{H}}^{\delta}|_{H_{\tilde{\mathbf{k}}}^{2}(\tilde{\Gamma})}$. Thus, there exists the corresponding eigenfunction $\tilde{u}^{\delta}\in H_{\tilde{\mathbf{k}}}^{2}(\tilde{\Gamma})$ satisfying (\ref{eq:bloch}) and
\begin{equation*}
    \tilde{\mathscr{H}}^{\delta}\tilde{u}^{\delta}(\tilde{\mathbf{r}})=\tilde{\lambda}^{\delta}\tilde{u}^{\delta}(\tilde{\mathbf{r}}), \quad \forall\tilde{\mathbf{r}}\in \tilde{\Gamma}. 
\end{equation*}

Next, we show that $\tilde{u}^{\delta}\in H_{\tilde{\mathbf{k}}}^{2}(\tilde{\Gamma})$ can be extended to $\mathbb{R}^{2d}$ and satisfy (\ref{eq:re_ehn}).

Let $T_{\tilde{\tau}} (\tilde{\tau}\in\tilde{\mathcal{R}})$ be translation operators deﬁned by
\begin{align*}
    T_{\tilde{\tau}}f(\tilde{\mathbf{r}}):=f(\tilde{\mathbf{r}}+\tilde{\tau}).
\end{align*}
Then due to the periodicity of the potential $V_{1}+V_{2}$, we see that $\tilde{\mathscr{H}}^{\delta}$ and $T_{\tilde{\tau}} (\tilde{\tau}\in\tilde{\mathcal{R}})$ commute among each other, i.e.,
\begin{align*}
    T_{\tilde{\tau}}\circ\tilde{\mathscr{H}}^{\delta}=\tilde{\mathscr{H}}^{\delta}\circ T_{\tilde{\tau}},\quad \forall\tilde{\tau}\in\tilde{\mathcal{R}}.
\end{align*}
For $\tilde{\mathbf{r}}\in\mathbb{R}^{2d}$, there exists $\tilde{\tau}\in\tilde{\mathcal{R}}$ such that $\tilde{\mathbf{r}}+\tilde{\tau}\in\tilde{\Gamma}$. Then we arrive at 
\begin{equation*}
    \tilde{\mathscr{H}}^{\delta}\tilde{u}^{\delta}(\tilde{\mathbf{r}})=T_{\tilde{\tau}}^{-1}\tilde{\mathscr{H}}^{\delta}T_{\tilde{\tau}}\tilde{u}^{\delta}(\tilde{\mathbf{r}})=T_{\tilde{\tau}}^{-1}\tilde{\mathscr{H}}^{\delta}\tilde{u}^{\delta}(\tilde{\mathbf{r}}+\tilde{\tau})=T_{\tilde{\tau}}^{-1}\tilde{\lambda}^{\delta}\tilde{u}^{\delta}(\tilde{\mathbf{r}}+\tilde{\tau})=\tilde{\lambda}^{\delta}\tilde{u}^{\delta}(\tilde{\mathbf{r}}),
\end{equation*}
which implies $\tilde{u}^{\delta}$ is the Bloch solution to (\ref{eq:re_ehn}).

If there exists $\tilde{\mathbf{k}}\in\tilde{\Gamma}^{*}$ and the Bloch solution $\tilde{u}^{\delta}\in H^{2}_{\mathrm{loc}}(\mathbb{R}^{2d})$ of (\ref{eq:re_ehn}) with some $\tilde{\lambda}^{\delta}\in\mathbb{R}$, then $\tilde{u}^{\delta}$ satisfies (\ref{eq:bloch}) and
\begin{align*}
    \tilde{u}^{\delta}\in H_{\tilde{\mathbf{k}}}^{2}(\tilde{\Gamma}).
\end{align*}
Therefore, $(\tilde{\lambda}^{\delta}, \tilde{u}^{\delta})$ is the eigenpair of $\tilde{\mathscr{H}}^{\delta}|_{H_{\tilde{\mathbf{k}}}^{2}(\tilde{\Gamma})}$. We conclude from Proposition \ref{prop:spec_dec} that
\begin{align*}
    \tilde{\lambda}^{\delta}\in \sigma(\tilde{\mathscr{H}}^{\delta}|_{H_{\tilde{\mathbf{k}}}^{2}(\tilde{\Gamma})})=\sigma(\tilde{\mathscr{H}}^{\delta}(\tilde{\mathbf{k}}))\subset\bigcup_{\tilde{\mathbf{k}}\in\tilde{\Gamma}^{*}}\sigma(\tilde{\mathscr{H}}^{\delta}(\tilde{\mathbf{k}}))=\sigma(\tilde{\mathscr{H}}^{\delta}).
\end{align*}
This completes the proof.
\end{proof}

\subsection{Relationship between $\sigma_{\mathcal{G}}(\tilde{\mathscr{H}})$ and $\sigma(\tilde{\mathscr{H}}^{\delta})$}
In this section, we study the relationship between the spectra of the extended Schrödinger operator $\tilde{\mathscr{H}}$ and the regularized Schrödinger operator $\tilde{\mathscr{H}}^{\delta}$.
The following lemma plays a crucial role in our analysis (see  \cite{de2008intermediate}).
\begin{lemma}\label{lemma:limit}
Let $\mathscr{L}_{n} (n\geqslant1)$ and $\mathscr{L}$ be self-adjoint operators acting on the Hilbert space $H$. If there holds for the corresponding resolvents that
\begin{align*}
    R(\mathrm{i}; \mathscr{L}_{n})\xrightarrow{s}R(\mathrm{i}; \mathscr{L}),
\end{align*}
and $\lambda\in\sigma(\mathscr{L})$,
then there exists a sequence $\{\lambda_{n}\}_{n\geqslant1}$ with $\lambda_{n}\in\sigma(\mathscr{L}_{n})$ such that
\begin{align*}
    \lambda_{n}\rightarrow\lambda,
\end{align*}
where 
\begin{align*}
     R(\mathrm{i}, \mathscr{L}_{n}):=(\mathscr{L}_{n}-\mathrm{i}I)^{-1},\quad R(\mathrm{i}, \mathscr{L}):=(\mathscr{L}-\mathrm{i}I)^{-1}.
\end{align*}
\end{lemma}

The following theorem shows that for any $\tilde{\mathbf{k}}\in \tilde{\Gamma}^{*}$, the spectrum of the degenerate elliptic operator $\tilde{\mathscr{H}}(\tilde{\mathbf{k}})$ on the domain $\mathcal{G}(\mathbb{T}^{2d})$ can be approximated by the spectra of a family of regularized elliptic operators $\left\{\tilde{\mathscr{H}}^{\delta}(\tilde{\mathbf{k}})\right\}_{\delta>0}$ on the domain $H^{2}(\mathbb{T}^{2d})$.

\begin{theorem}\label{thm:spectttt}
For any $\tilde{\mathbf{k}}\in \tilde{\Gamma}^{*}$ and 
$\tilde{\lambda}\in\sigma_{\mathcal{G}}(\tilde{\mathscr{H}}(\tilde{\mathbf{k}}))$, there exists  $\tilde{\lambda}^{\delta}\in\sigma(\tilde{\mathscr{H}}^{\delta}(\tilde{\mathbf{k}}))$ such that
\begin{align*}
\tilde{\lambda}=\lim_{\delta\rightarrow0^{+}}\tilde{\lambda}^{\delta}.
\end{align*}
\end{theorem}
\begin{proof}
Since both $\tilde{\mathscr{H}}(\tilde{\mathbf{k}})|_{\mathcal{G}(\mathbb{T}^{2d})}$ and $\tilde{\mathscr{H}}(\tilde{\mathbf{k}})^{\delta}|_{H^{2}(\mathbb{T}^{2d})}$ are self-adjoint for any $\tilde{\mathbf{k}}\in\tilde{\Gamma}^{*}$, 
$\sigma_{\mathcal{G}}(\tilde{\mathscr{H}}(\tilde{\mathbf{k}}))\subset\mathbb{R}$ and
\begin{align*}
    \operatorname{dist}(\mathrm{i}, \sigma_{\mathcal{G}}(\tilde{\mathscr{H}}(\tilde{\mathbf{k}})))=\inf_{\lambda\in\sigma_{\mathcal{G}}(\tilde{\mathscr{H}}(\tilde{\mathbf{k}}))}\operatorname{dist}(\mathrm{i}, \sigma_{\mathcal{G}}(\tilde{\mathscr{H}}(\tilde{\mathbf{k}})))\geqslant1,\quad \forall \tilde{\mathbf{k}}\in\tilde{\Gamma}^{*},
\end{align*}
which implies that
\begin{align}\label{ine:modh}
    \left\Vert R(\mathrm{i}; \tilde{\mathscr{H}}(\tilde{\mathbf{k}})) \right\Vert_{L^{2}(\mathbb{T}^{2d})\rightarrow L^{2}(\mathbb{T}^{2d})}\leqslant\frac{1}{ \operatorname{dist}(\mathrm{i}, \sigma_{\mathcal{G}}(\tilde{\mathscr{H}}(\tilde{\mathbf{k}})))}\leqslant1, \quad \forall \tilde{\mathbf{k}}\in\tilde{\Gamma}^{*}.
\end{align}
Similarly, we obtain
\begin{align}\label{ine:modhe}
    \left\Vert R(\mathrm{i}; \tilde{\mathscr{H}}^{\delta}(\tilde{\mathbf{k}})) \right\Vert_{L^{2}(\mathbb{T}^{2d})\rightarrow L^{2}(\mathbb{T}^{2d})}\leqslant1, \quad \forall \tilde{\mathbf{k}}\in\tilde{\Gamma}^{*}.
\end{align}

Note that for $f\in C^{\infty}(\mathbb{T}^{2d})$ and $\tilde{\mathbf{k}}\in\tilde{\Gamma}^{*}$, we have
\begin{align*}
    \left\Vert \tilde{\mathscr{H}}^{\delta}(\tilde{\mathbf{k}})f-\tilde{\mathscr{H}}(\tilde{\mathbf{k}})f\right\Vert_{L^{2}(\mathbb{T}^{2d})}=&\left\Vert\frac{\delta}{2}\left(\mathrm{i}\left(\nabla_{\mathbf{r}}-\nabla_{\mathbf{r}'}\right)-\left(\mathbf{k}-\mathbf{k}'\right)\right)^{2}f\right\Vert_{L^{2}(\mathbb{T}^{2d})}\leqslant C\delta\Vert f\Vert_{H^{2}(\mathbb{T}^{2d})},
\end{align*}
where the constant $C>0$ is independent of $\delta$ and $\tilde{\mathbf{k}}\in\tilde{\Gamma}^{*}$, which together with (\ref{ine:modhe}) and the fact
\begin{align*}
    R(\mathrm{i}; \tilde{\mathscr{H}}(\tilde{\mathbf{k}}))f\in C^{\infty}(\mathbb{T}^{2d}),\quad\forall\tilde{\mathbf{k}}\in\tilde{\Gamma}^{*}
\end{align*}
leads to
\begin{equation}
   \begin{aligned}\label{ine:1_2}
    &\left\Vert R(\mathrm{i}; \tilde{\mathscr{H}}(\tilde{\mathbf{k}}))f-R(\mathrm{i}; \tilde{\mathscr{H}}^{\delta}(\tilde{\mathbf{k}}))f\right\Vert_{L^{2}(\mathbb{T}^{2d})}\\=&\left\Vert R(\mathrm{i}; \tilde{\mathscr{H}}^{\delta}(\tilde{\mathbf{k}}))\left(\tilde{\mathscr{H}}^{\delta}(\tilde{\mathbf{k}})-\tilde{\mathscr{H}}(\tilde{\mathbf{k}})\right)R(\mathrm{i}; \tilde{\mathscr{H}}(\tilde{\mathbf{k}}))f\right\Vert_{L^{2}(\mathbb{T}^{2d})}\\\leqslant&\left\Vert R(\mathrm{i}; \tilde{\mathscr{H}}^{\delta}(\tilde{\mathbf{k}}))\right\Vert_{L^{2}(\mathbb{T}^{2d})\rightarrow L^{2}(\mathbb{T}^{2d})}\left\Vert\left(\tilde{\mathscr{H}}^{\delta}(\tilde{\mathbf{k}})-\tilde{\mathscr{H}}(\tilde{\mathbf{k}})\right)R(\mathrm{i}; \tilde{\mathscr{H}}(\tilde{\mathbf{k}}))f\right\Vert_{L^{2}(\mathbb{T}^{2d})}
    \\\leqslant& C\delta \left\Vert R(\mathrm{i}; \tilde{\mathscr{H}}(\tilde{\mathbf{k}}))f\right\Vert_{H^{2}(\mathbb{T}^{2d})}.
\end{aligned} 
\end{equation}

We see that $C^{\infty}(\mathbb{T}^{2d})$ is dense in $L^{2}(\mathbb{T}^{2d})$. Then, for any $g\in L^{2}(\mathbb{T}^{2d})$, there exists a sequence $\{g_{n}\}_{n\geqslant1}\subset C^{\infty}(\mathbb{T}^{2d})$ such that
\begin{align*}
    \Vert g_{n}-g\Vert_{L^{2}(\mathbb{T}^{2d})}\rightarrow0,\quad \text{as}\quad n\rightarrow\infty,
\end{align*}
which together with (\ref{ine:modh}), (\ref{ine:modhe}) and (\ref{ine:1_2}) implies that 
\begin{align*}
    &\left\Vert R(\mathrm{i}; \tilde{\mathscr{H}}(\tilde{\mathbf{k}}))g-R(\mathrm{i}; \tilde{\mathscr{H}}^{\delta}(\tilde{\mathbf{k}}))g\right\Vert_{L^{2}(\mathbb{T}^{2d})}\\\leqslant&\left\Vert \left(R(\mathrm{i}; \tilde{\mathscr{H}}(\tilde{\mathbf{k}}))-R(\mathrm{i}; \tilde{\mathscr{H}}^{\delta}(\tilde{\mathbf{k}}))\right)(g-g_{n})\right\Vert_{L^{2}(\mathbb{T}^{2d})}\\&+\left\Vert \left(R(\mathrm{i}; \tilde{\mathscr{H}}(\tilde{\mathbf{k}}))-R(\mathrm{i}; \tilde{\mathscr{H}}^{\delta}(\tilde{\mathbf{k}}))\right)g_{n}\right\Vert_{L^{2}(\mathbb{T}^{2d})}\\\leqslant&2\Vert g_{n}-g\Vert_{L^{2}(\mathbb{T}^{2d})}+C\delta \left\Vert R(\mathrm{i}; \tilde{\mathscr{H}}(\tilde{\mathbf{k}}))g_{n}\right\Vert_{H^{2}(\mathbb{T}^{2d})}\rightarrow0,
\end{align*}
as $\delta\rightarrow0^{+}$ and $n\rightarrow\infty$.
That is, 
\begin{align*}
    R(\mathrm{i}; \tilde{\mathscr{H}}(\tilde{\mathbf{k}})^{\delta})\xrightarrow{s}R(\mathrm{i}; \tilde{\mathscr{H}}(\tilde{\mathbf{k}})),\quad \text{as}\quad \delta\rightarrow0^{+}, \quad\forall \tilde{\mathbf{k}}\in\tilde{\Gamma}^{*},
\end{align*}
which together with Lemma \ref{lemma:limit} completes the proof.
\end{proof}

Theorem \ref{thm:spectttt} tells that $\sigma_{\mathcal{G}}(\tilde{\mathscr{H}}(\tilde{\mathbf{k}}))$ can be approximated by $\sigma(\tilde{\mathscr{H}}^{\delta}(\tilde{\mathbf{k}}))$ as $\delta\rightarrow0^{+}$ for $\tilde{\mathbf{k}}\in\tilde{\Gamma}^{*}$, which together with the relationship between $\sigma_{\mathcal{G}}(\tilde{\mathscr{H}})$ and $\bigcup_{\tilde{\mathbf{k}}\in\tilde{\Gamma}^{*}}\sigma_{\mathcal{G}}(\tilde{\mathscr{H}}(\tilde{\mathbf{k}}))$ shown as in Theorem \ref{prop:dense} yields the following corollary.
\begin{corollary}\label{cor:dens}
For $\tilde{\lambda}\in\sigma_{\mathcal{G}}(\tilde{\mathscr{H}})$, there exists $\tilde{\lambda}^{\delta}\in\sigma(\tilde{\mathscr{H}}^{\delta})$ such that
\begin{align*}
\tilde{\lambda}=\lim_{\delta\rightarrow0^{+}}\tilde{\lambda}^{\delta}.
\end{align*}
\end{corollary}
\begin{proof}
For $\tilde{\lambda}\in\sigma_{\mathcal{G}}(\tilde{\mathscr{H}})$, we obtain from Theorem \ref{prop:dense} that for $\varepsilon>0$, there exists $\tilde{\mathbf{k}}_{\varepsilon}\in \tilde{\Gamma}^{*}$ and $\tilde{\lambda}_{\varepsilon}\in\sigma(\tilde{\mathscr{H}}(\tilde{\mathbf{k}}_{\varepsilon}))$ such that
\begin{align*}
    \left|\tilde{\lambda}-\tilde{\lambda}_{\varepsilon}\right|<\varepsilon.
\end{align*}
For $\tilde{\lambda}_{\varepsilon}\in\sigma(\tilde{\mathscr{H}}(\tilde{\mathbf{k}}_{\varepsilon}))$, we see from Theorem \ref{thm:spectttt} that there exists $\tilde{\lambda}_{\varepsilon}^{\delta}\in\sigma(\tilde{\mathscr{H}}^{\delta}(\tilde{\mathbf{k}}_{\varepsilon}))$ for $\delta>0$ satisfying
\begin{align*}
    \tilde{\lambda}_{\varepsilon}=\lim_{\delta\rightarrow0^{+}}\tilde{\lambda}_{\varepsilon}^{\delta}.
\end{align*}
Consider the family $\{\tilde{\lambda}_{\delta}^{\delta}\}_{\delta>0}\subset\bigcup_{\tilde{\mathbf{k}}\in\tilde{\Gamma}^{*}}\sigma(\tilde{\mathscr{H}}^{\delta}(\tilde{\mathbf{k}}))=\sigma(\tilde{\mathscr{H}}^{\delta})$. If we set $\tilde{\lambda}^{\delta}:=\tilde{\lambda}_{\delta}^{\delta}$ for $\delta>0$, then
\begin{align*}
\tilde{\lambda}=\lim_{\delta\rightarrow0^{+}}\tilde{\lambda}^{\delta}.
\end{align*}
This completes the proof.
\end{proof}
\section{Spectrum of the Schrödinger operator $\mathscr{H}$}

Building on our previous analysis of the extended Schrödinger operator $\tilde{\mathscr{H}}$, together with the introduction of a regularization technique and the analysis of the corresponding regularized operator $\tilde{\mathscr{H}}^{\delta}$, we now return to the spectrum of the original Schrödinger operator $\mathscr{H}$ for the incommensurate system.
\subsection{Relationship between $\sigma(\mathscr{H})$ and $\sigma(\tilde{\mathscr{H}}^{\delta})$}
In this subsection, we study the relationship between the spectra of the Schrödinger operator $\mathscr{H}$ for the incommensurate system and the regularized Schrödinger operator $\tilde{\mathscr{H}}^{\delta}$. We first establish the relationship between the spectrum of $\mathscr{H}$ and the spectrum of the extended Schrödinger operator $\tilde{\mathscr{H}}$. The lemma below records the standard spectral properties of self-adjoint operators \cite{reed1980methods} and will be used in our analysis.
 
\begin{lemma}\label{lem:app_p}
Let $\mathcal{L}$ be a self-adjoint operator in the Hilbert space $H$ on the domain $\mathcal{D}$. Then, $\lambda\in \sigma(\mathcal{L})$ is equivalent to the fact that there exists a sequence $\{\varphi_{n}\}_{n=1}^{\infty}\subset\mathcal{D}$ such that $\Vert \varphi_{n}\Vert=1$ and
\begin{align*}
    \left\Vert(\mathcal{L}-\lambda I)\varphi_{n}\right\Vert\rightarrow0,\quad \text{as}\quad n\rightarrow\infty.
\end{align*}
\end{lemma}

The following conclusion shows the spectral consistency between the Schrödinger operator $\mathscr{H}$ and its associated extended Schrödinger operator $\tilde{\mathscr{H}}$ after the self-adjoint extension.

\begin{theorem}\label{prop:susbs}
There hold that
\begin{enumerate}[label=(\alph*)]
\item  
\begin{align*}
\sigma(\mathscr{H})\subset\sigma_{\mathcal{G}}(\tilde{\mathscr{H}});
\end{align*}
\item if $\tilde{\lambda}\in \sigma_{\mathcal{G}}(\tilde{\mathscr{H}})$ satisfies the generalized eigenfunction condition, namely, there exists $\tilde{\mathbf{k}}_{*}\in \tilde{\Gamma}^{*}$ and $0\neq\tilde{v}_{*}\in \mathcal{G}(\mathbb{T}^{2d})$ satisfying  
\begin{align}\label{eq:thk}
    \tilde{\mathscr{H}}(\tilde{\mathbf{k}}_{*})\tilde{v}_{*}=\tilde{\lambda}\tilde{v}_{*},
\end{align}
then 
\begin{align*}
    \tilde{\lambda}\in \sigma(\mathscr{H}).
\end{align*}
\end{enumerate}
\end{theorem}
\begin{proof}
(a) Consider the following extended Schrödinger operator:
\begin{align*}
    \tilde{\mathscr{H}}:=-\frac{1}{2} \sum_{i=1}^{d}(\partial_{\mathbf{r}_{i}}+\partial_{\mathbf{r}_{i}'})^{2}+V_{1}(\mathbf{r})+V_{2}(\mathbf{r}'),
\end{align*}
where $\tilde{\mathbf{r}}:=(\mathbf{r}, \mathbf{r}')\in\mathbb{R}^{d}\times\mathbb{R}^{d}$ and $\mathbf{r}=(\mathbf{r}_{1},\ldots,\mathbf{r}_{d})$, $\mathbf{r}'=(\mathbf{r}_{1}',\ldots,\mathbf{r}_{d}')$.

Let
\begin{align*}
    \mathbf{p}_{i}=\frac{1}{2}\left(\mathbf{r}_{i}+\mathbf{r}_{i}'\right),\quad \mathbf{q}_{i}=\frac{1}{2}\left(\mathbf{r}_{i}-\mathbf{r}_{i}'\right),\quad i=1,\ldots,d,
\end{align*}
and we rewrite the extended Schrödinger operator $\tilde{\mathscr{H}}$ as
\begin{align*}
    \tilde{\mathscr{H}}_{\mathbf{p}\mathbf{q}}:=-\frac{1}{2}\sum_{i=1}^{d}\partial_{\mathbf{p}_{i}}^{2}+V_{1}(\mathbf{p}+\mathbf{q})+V_{2}(\mathbf{p}-\mathbf{q}),
\end{align*}
where $\mathbf{p}=(\mathbf{p}_{1}, \ldots, \mathbf{p}_{d})$ and $\mathbf{q}=(\mathbf{q}_{1}, \ldots, \mathbf{q}_{d})$. It is observed that 
\begin{align*}
    \tilde{\mathscr{H}}_{\mathbf{p}\mathbf{q}}|_{\mathbf{p}=\mathbf{r}, \mathbf{q}=\mathbf{0}}=\mathscr{H}.
\end{align*}

If $\lambda\in \sigma(\mathscr{H})$, then by Lemma \ref{lem:app_p}, there exists a sequence $\{\varphi_{n}\}_{n=1}^{\infty}\subset H^{2}(\mathbb{R}^{d})$ such that $\Vert \varphi_{n}\Vert_{L^{2}(\mathbb{R}^{d})}=1$ and
\begin{align}\label{lim:low}
    \left\Vert(\mathscr{H}-\lambda I)\varphi_{n}\right\Vert_{L^{2}(\mathbb{R}^{d})}\rightarrow0,\quad \text{as}\quad n\rightarrow\infty.
\end{align}
Define a sequence $\{\psi_{n}\}_{n=1}^{\infty}\subset \mathcal{G}$ by
\begin{align*}
    \psi_{n}(\mathbf{r}, \mathbf{r}'):=\varphi_{n}(\frac{\mathbf{r}+\mathbf{r}'}{2})\eta_{n}(\frac{\mathbf{r}-\mathbf{r}'}{2}),\quad (\mathbf{r}, \mathbf{r}')\in \mathbb{R}^{2d},
\end{align*}
where the sequence $\{\eta_{n}\}_{n=1}^{\infty}\subset C^{\infty}(\mathbb{R}^{2d})$ centered at $\mathbf{0}$ with $\Vert \eta_{n}\Vert_{L^{2}(\mathbb{R}^{d})}=1$ converges to the Dirac delta $\delta(\cdot)$ as $n\rightarrow\infty$. A typical example for $\{\eta_{n}\}_{n=1}^{\infty}$ is the Gaussian functions. 

We observe that 
\begin{align*}
    \left\Vert\psi_{n}(\mathbf{r}, \mathbf{r}')\right\Vert_{L^{2}(\mathbb{R}^{2d})}=&\left\Vert\varphi_{n}(\frac{\mathbf{r}+\mathbf{r}'}{2})\eta_{n}(\frac{\mathbf{r}-\mathbf{r}'}{2})\right\Vert_{L^{2}(\mathbb{R}^{2d})}=\left|\frac{\partial(\mathbf{p,q})}{\partial(\mathbf{r,r'})}\right|\left\Vert\varphi_{n}(\mathbf{p})\eta_{n}(\mathbf{q})\right\Vert_{L^{2}(\mathbb{R}^{2d})}\\=&\left|\frac{\partial(\mathbf{p,q})}{\partial(\mathbf{r,r'})}\right|\left\Vert\varphi_{n}\right\Vert_{L^{2}(\mathbb{R}^{d})}\left\Vert\eta_{n}\right\Vert_{L^{2}(\mathbb{R}^{d})}=\left(\frac{1}{2}\right)^{2d-1},\quad n=1, 2, \ldots.
\end{align*}

Note that 
    \begin{align*}
     \left(\tilde{\mathscr{H}}\psi_{n}\right)(\mathbf{r}, \mathbf{r}')=&\left(-\frac{1}{2} \sum_{i=1}^{d}(\partial_{\mathbf{r}_{i}}+\partial_{\mathbf{r}_{i}'})^{2}+V_{1}(\mathbf{r})+V_{2}(\mathbf{r}')\right)\varphi_{n}(\frac{\mathbf{r}+\mathbf{r}'}{2})\eta_{n}(\frac{\mathbf{r}-\mathbf{r}'}{2})\\=&\left(-\frac{1}{2}\sum_{i=1}^{d}\partial_{\mathbf{p}_{i}}^{2}+V_{1}(\mathbf{p}+\mathbf{q})+V_{2}(\mathbf{p}-\mathbf{q})\right)\varphi_{n}(\mathbf{p})\eta_{n}(\mathbf{q}),
        \end{align*}
which yields that
\begin{equation}
    \begin{aligned}\label{eq:sum}
    &\left(\tilde{\mathscr{H}}\psi_{n}\right)(\mathbf{r}, \mathbf{r}')-\lambda\psi_{n}(\mathbf{r}, \mathbf{r}')\\=&-\frac{1}{2}\sum_{i=1}^{d}\partial_{\mathbf{p}_{i}}^{2}\varphi_{n}(\mathbf{p})\eta_{n}(\mathbf{q})+\left(V_{1}(\mathbf{p}+\mathbf{q})+V_{2}(\mathbf{p}-\mathbf{q})-\lambda\right)\varphi_{n}(\mathbf{p})\eta_{n}(\mathbf{q})\\=&
    \left(\mathscr{H}-\lambda I\right)\varphi_{n}(\mathbf{p})\eta_{n}(\mathbf{q})+\left(V_{1}(\mathbf{p}+\mathbf{q})+V_{2}(\mathbf{p}-\mathbf{q})-V_{1}(\mathbf{p})-V_{2}(\mathbf{p})\right)\varphi_{n}(\mathbf{p})\eta_{n}(\mathbf{q}).
\end{aligned}
\end{equation}

For the first term, we obtain from (\ref{lim:low}) that
\begin{equation}
    \begin{aligned}\label{eq:first_term}
    &\left\Vert \left(\mathscr{H}-\lambda I\right)\varphi_{n}(\mathbf{p})\eta_{n}(\mathbf{q})\right\Vert_{L^{2}(\mathbb{R}^{2d})}\\=&\left\Vert \left(\mathscr{H}-\lambda I\right)\varphi_{n}(\mathbf{p})\right\Vert_{L^{2}(\mathbb{R}^{d})}\left\Vert\eta_{n}(\mathbf{q})\right\Vert_{L^{2}(\mathbb{R}^{d})}=\left\Vert \left(\mathscr{H}-\lambda I\right)\varphi_{n}\right\Vert_{L^{2}(\mathbb{R}^{d})}\rightarrow0,
\end{aligned}
\end{equation}
as $n\rightarrow\infty$. 

For the second term, we may conduct as follows
\begin{equation}
    \begin{aligned}\label{eq:sec_term}
    &\left\Vert\left(V_{1}(\mathbf{p}+\mathbf{q})+V_{2}(\mathbf{p}-\mathbf{q})-V_{1}(\mathbf{p})-V_{2}(\mathbf{p})\right)\varphi_{n}(\mathbf{p})\eta_{n}(\mathbf{q})\right\Vert_{L^{2}(\mathbb{R}^{2d})}^{2}\\=&\int_{\Vert\mathbf{q}\Vert>\gamma}\left|\left(V_{1}(\mathbf{p}+\mathbf{q})+V_{2}(\mathbf{p}-\mathbf{q})-V_{1}(\mathbf{p})-V_{2}(\mathbf{p})\right)\varphi_{n}(\mathbf{p})\eta_{n}(\mathbf{q})\right|^{2}d\mathbf{p}d\mathbf{q}\\&+\int_{\Vert\mathbf{q}\Vert\leqslant\gamma}\left|\left(V_{1}(\mathbf{p}+\mathbf{q})+V_{2}(\mathbf{p}-\mathbf{q})-V_{1}(\mathbf{p})-V_{2}(\mathbf{p})\right)\varphi_{n}(\mathbf{p})\eta_{n}(\mathbf{q})\right|^{2}d\mathbf{p}d\mathbf{q}\\\leqslant&4\left(\Vert V_{1}\Vert_{L^{\infty}}+\Vert V_{2}\Vert_{L^{\infty}}\right)\left\Vert\varphi_{n}\right\Vert_{L^{2}(\mathbb{R}^{d})}^{2}\int_{\Vert\mathbf{q}\Vert>\gamma}\left|\eta_{n}(\mathbf{q})\right|^{2}d\mathbf{q}\\&+\max_{\mathbf{p}, \Vert\mathbf{q}\Vert\leqslant\gamma}\left|\left(V_{1}(\mathbf{p}+\mathbf{q})+V_{2}(\mathbf{p}-\mathbf{q})-V_{1}(\mathbf{p})-V_{2}(\mathbf{p})\right)\right|^{2}\left\Vert\varphi_{n}\right\Vert_{L^{2}(\mathbb{R}^{d})}^{2}\int_{\Vert\mathbf{q}\Vert\leqslant\gamma}\left|\eta_{n}(\mathbf{q})\right|^{2}d\mathbf{q}.
\end{aligned}
\end{equation}
Since $\{\eta_{n}\}_{n=1}^{\infty}\subset C^{\infty}(\mathbb{R}^{2d})$ centered at $\mathbf{0}$, we have 
\begin{align*}
    \int_{\Vert\mathbf{q}\Vert>\gamma}\left|\eta_{n}(\mathbf{q})\right|^{2}d\mathbf{q}\rightarrow0,\quad \text{as}\quad n\rightarrow\infty,\quad\forall\gamma>0.
\end{align*}
By the uniform continuity and the periodicity of $V_{j} (j=1,2)$, there holds that
\begin{align*}
    \max_{\mathbf{p}, \Vert\mathbf{q}\Vert\leqslant\gamma}\left|\left(V_{1}(\mathbf{p}+\mathbf{q})+V_{2}(\mathbf{p}-\mathbf{q})-V_{1}(\mathbf{p})-V_{2}(\mathbf{p})\right)\right|\rightarrow0,  \quad\text{as}\quad \gamma\rightarrow0^{+}.
\end{align*}
Then, we obtain that 
\begin{align*}
    \left\Vert\left(V_{1}(\mathbf{p}+\mathbf{q})+V_{2}(\mathbf{p}-\mathbf{q})-V_{1}(\mathbf{p})-V_{2}(\mathbf{p})\right)\varphi_{n}(\mathbf{p})\eta_{n}(\mathbf{q})\right\Vert_{L^{2}(\mathbb{R}^{2d})}^{2}\rightarrow0,\quad \text{as}\quad n\rightarrow\infty,
\end{align*}
which together with  (\ref{eq:sum}), (\ref{eq:first_term}) and (\ref{eq:sec_term}) produces 
\begin{align*}
\left\Vert\tilde{\mathscr{H}}\psi_{n}-\lambda\psi_{n}\right\Vert_{L^{2}(\mathbb{R}^{2d})}\rightarrow0,\quad \text{as}\quad n\rightarrow\infty.
\end{align*}
Therefore, we derive from Lemma \ref{lem:app_p}  that $\lambda\in\sigma_{\mathcal{G}}(\tilde{\mathscr{H}})$ and then $\sigma(\mathscr{H})\subset\sigma_{\mathcal{G}}(\tilde{\mathscr{H}})$.

\hspace{10pt}

(b)
Let $\tilde{\lambda}\in\sigma_{\mathcal{G}}(\tilde{\mathscr{H}})$ and suppose there exists $\tilde{\mathbf{k}}_{*}\in \tilde{\Gamma}^{*}$ and $0\neq\tilde{v}_{*}\in \mathcal{G}(\mathbb{T}^{2d})$ satisfying (\ref{eq:thk}).

For $\mathbf{p}=(\mathbf{p}_{1}, \ldots, \mathbf{p}_{d})$, $\mathbf{q}=(\mathbf{q}_{1}, \ldots, \mathbf{q}_{d})$ and $(\mathbf{p}, \mathbf{q})\in \mathbb{T}_{\mathbf{pq}}^{2d}$, we set
\begin{align*}
    \tilde{v}_{*}(\tilde{\mathbf{r}})= \tilde{v}_{*}(\mathbf{r}, \mathbf{r}')=\tilde{v}_{*}(\mathbf{p}+\mathbf{q}, \mathbf{p}-\mathbf{q})\triangleq\tilde{\textbf{v}}_{*}(\mathbf{p},\mathbf{q}).
\end{align*}
Similarly, for $\tilde{u}_{*}(\tilde{\mathbf{r}})=e^{\mathrm{i}\tilde{\mathbf{k}}_{*}\cdot\tilde{\mathbf{r}}}\tilde{v}_{*}(\tilde{\mathbf{r}})$, denote
\begin{align*}
    \tilde{u}_{*}(\tilde{\mathbf{r}})\triangleq\tilde{\textbf{u}}_{*}(\mathbf{p},\mathbf{q}).
\end{align*}
By simple calculations, we see that 
\begin{align}\label{eq:hhh}
\tilde{\mathscr{H}}_{\mathbf{p}\mathbf{q}}\tilde{\textbf{u}}_{*}=\tilde{\lambda}\tilde{\textbf{u}}_{*}.
\end{align}

Note that the assumption $\tilde{v}_{*}\in \mathcal{G}(\mathbb{T}^{2d})$, together with $V_{j}\in C^{0}(\Gamma_{j}) (j=1,2)$, implies that
\begin{align*}
\infty>&\int_{\mathbb{T}^{2d}}\left|\tilde{\mathscr{H}}\tilde{v}_{*}\right|^{2}d\tilde{\mathbf{r}}\geqslant\frac{1}{8}\int_{\mathbb{T}^{2d}}\left|\sum_{i=1}^{d}(\partial_{\mathbf{r}_{i}}+\partial_{\mathbf{r}_{i}'})^{2}\tilde{v}_{*}\right|^{2}d\tilde{\mathbf{r}}-\int_{\mathbb{T}^{2d}}\left|\left(V_{1}(\mathbf{r})+V_{2}(\mathbf{r}')\right)\tilde{v}_{*}\right|^{2}d\tilde{\mathbf{r}}\\\geqslant&\frac{1}{8}\int_{\mathbb{T}^{2d}}\left|\sum_{i=1}^{d}(\partial_{\mathbf{r}_{i}}+\partial_{\mathbf{r}_{i}'})^{2}\tilde{v}_{*}\right|^{2}d\tilde{\mathbf{r}}-\left(\Vert V_{1}\Vert_{L^{\infty}}+\Vert V_{2}\Vert_{L^{\infty}}\right)^{2}\left\Vert\tilde{v}_{*} \right\Vert^{2}_{L^{2}(\mathbb{T}^{2d})}.
\end{align*}
Then, we have
\begin{align*}
\int_{\mathbb{T}_{\mathbf{pq}}^{2d}}\left|\sum_{i=1}^{d}\partial_{\mathbf{p}_{i}}^{2}\tilde{\textbf{v}}_{*}\right|^{2}<\infty,
\end{align*}
which together with Fubini's theorem yields that for almost all $\mathbf{q}$, there hold
\begin{align*}
    \tilde{\textbf{v}}_{*}(\cdot,\mathbf{q})\in H^{2}_{\mathrm{loc}}(\mathbb{R}^{d}),\quad\text{and}\quad\tilde{\textbf{u}}_{*}(\cdot,\mathbf{q})\in H^{2}_{\mathrm{loc}}(\mathbb{R}^{d}).
\end{align*}
Due to the fact that $\tilde{v}_{*}\neq0$, by Fubini's theorem again, there exists  $\mathbf{q_{0}}$ such that 
\begin{align*}
    0\neq\tilde{\textbf{u}}_{*}(\cdot,\mathbf{q_{0}})\in H^{2}_{\mathrm{loc}}(\mathbb{R}^{d}).
\end{align*}
Note that $\mathbf{q_{0}}$ is not guaranteed to be $\mathbf{0}$ here.

Denote the phase by
\begin{align*}
   \mathbf{\tilde{q}_{0}}=(\mathbf{q_{0}}, -\mathbf{q_{0}}),
\end{align*}
and for $\mathbf{r}\in\mathbb{R}^{2d}$,
\begin{align*}
    u_{*}(\mathbf{r}):=& \tilde{u}_{*}(\mathbf{r}+\mathbf{q_{0}}, \mathbf{r}-\mathbf{q_{0}})=\tilde{\textbf{u}}_{*}(\mathbf{r},\mathbf{q_{0}}), \\
    v_{*}(\mathbf{r}):=& \tilde{v}_{*}(\mathbf{r}+\mathbf{q_{0}}, \mathbf{r}-\mathbf{q_{0}})=\tilde{\textbf{v}}_{*}(\mathbf{r},\mathbf{q_{0}}). 
\end{align*}
Thus, $0\neq u_{*}\in H^{2}_{\mathrm{loc}}(\mathbb{R}^{d})$. By (\ref{eq:hhh}), there holds that
\begin{align*}
    \mathscr{H}_{\mathbf{\tilde{q}_{0}}}u_{*}:=-\frac{1}{2}\Delta u_{*}+\left(V_{1}(\mathbf{r}+\mathbf{q_{0}})+V_{2}(\mathbf{r}-\mathbf{q_{0}})\right)u_{*}=\tilde{\lambda}u_{*}.
\end{align*}

Next, we prove that $\tilde{\lambda}\in \sigma(\mathscr{H}_{\mathbf{\tilde{q}_{0}}})$. Let $\chi\in C_{c}^{\infty}(\mathbb{R})$ be the truncation function
\begin{equation*}
    \chi(t)=\left\{\begin{array}{ll}
         1, \quad t\leqslant0,\\
         0, \quad t\geqslant1.\\
    \end{array}\right.
\end{equation*}
Given $R>0$, define the radial truncation
\begin{align*}
     \chi_{R}(\mathbf{r}):=\chi(|\mathbf{r}|-R),\quad \mathbf{r}\in\mathbb{R}^{d}.
\end{align*}
Obviously, 
\begin{align*}
    \operatorname{supp}(\chi_{R})\subset B_{R}:=&\{\mathbf{r}: |\mathbf{r}|\leqslant R\}, \\\operatorname{supp}(\nabla \chi_{R}), \operatorname{supp}(\Delta \chi_{R})\subset S_{R}:=&\{\mathbf{r}: R\leqslant|\mathbf{r}|\leqslant R+1\}.
\end{align*}

If we set $\psi_{R}:=\chi_{R}u_{*}$, then 
\begin{align*}
    \mathscr{H}_{\mathbf{\tilde{q}_{0}}}\psi_{R}-\tilde{\lambda}\psi_{R}=&\left(\mathscr{H}_{\mathbf{\tilde{q}_{0}}}-\tilde{\lambda}I\right)\left(\chi_{R}u_{*}\right)\\=&\chi_{R}\left(\mathscr{H}_{\mathbf{\tilde{q}_{0}}}-\tilde{\lambda}I\right)u_{*}+[\mathscr{H}_{\mathbf{\tilde{q}_{0}}}, \chi_{R}]u_{*}=[\mathscr{H}_{\mathbf{\tilde{q}_{0}}}, \chi_{R}]u_{*}=[\Delta, \chi_{R}]u_{*}.
\end{align*}
After simple calculations, we have
\begin{align*}
    \left\Vert \mathscr{H}_{\mathbf{\tilde{q}_{0}}}\psi_{R}-\tilde{\lambda}\psi_{R}\right\Vert_{L^{2}(\mathbb{R}^{d})}^{2}= \left\Vert [\Delta, \chi_{R}]u_{*}\right\Vert_{L^{2}(\mathbb{R}^{d})}^{2}\leqslant C_{*}\int_{S_{R}}\left(\left|u_{*}\right|+\left|\nabla u_{*}\right|+\left|\Delta u_{*}\right|\right)^{2}d\mathbf{r},
\end{align*}
where $C_{*}>0$ is some constant that is independent of $R$. Note that
\begin{align*}
    u_{*}(\mathbf{r})=e^{\mathrm{i}\tilde{\mathbf{k}}_{*}\cdot(\mathbf{r}+\mathbf{q_{0}}, \mathbf{r}-\mathbf{q_{0}})}v_{*}(\mathbf{r}).
\end{align*}
 We obtain
\begin{align}\label{ine:s_r}
     \left\Vert \mathscr{H}_{\mathbf{\tilde{q}_{0}}}\psi_{R}-\tilde{\lambda}\psi_{R}\right\Vert_{L^{2}(\mathbb{R}^{d})}^{2}\leqslant \tilde{C}_{*}\int_{S_{R}}\left(\left|v_{*}\right|+\left|\nabla v_{*}\right|+\left|\Delta v_{*}\right|\right)^{2}d\mathbf{r},
\end{align}
where $\tilde{C}_{*}>0$ is some constant that is independent of $R$.

Since two Bravais lattices $\mathcal{R}_{1}$ and $\mathcal{R}_{2}$ are incommensurate, the identities below follow from the unique ergodicity of linear flows on tori (see, e.g., \cite{ding2010statistical},  \cite{walters2000introduction})
\begin{align*}
    \lim_{R\rightarrow\infty}\frac{1}{|B_{R}|}\int_{B_{R}}\left(\left|v_{*}\right|+\left|\nabla v_{*}\right|+\left|\Delta v_{*}\right|\right)^{2}d\mathbf{r}=&\frac{1}{\operatorname{Vol}(\mathbb{T}^{2d})}\int_{\mathbb{T}^{2d}}\left(\left|\tilde{v}_{*}\right|+\left|\nabla \tilde{v}_{*}\right|+\left|\Delta \tilde{v}_{*}\right|\right)^{2}d\tilde{\mathbf{r}}\triangleq J_{1},\\
    \lim_{R\rightarrow\infty}\frac{1}{|B_{R}|}\int_{B_{R}}\left|v_{*}\right|^{2}d\mathbf{r}=&\frac{1}{\operatorname{Vol}(\mathbb{T}^{2d})}\int_{\mathbb{T}^{2d}}\left|\tilde{v}_{*}\right|^{2}d\tilde{\mathbf{r}}\triangleq J_{2},
\end{align*}
which leads to that for $\varepsilon>0$ and $R\gg1$, there holds that
\begin{equation}
    \begin{aligned}\label{ine:R_e}
    |B_{R}|(J_{1}-\varepsilon)\leqslant&\int_{B_{R}}\left(\left|v_{*}\right|+\left|\nabla v_{*}\right|+\left|\Delta v_{*}\right|\right)^{2}d\mathbf{r}\leqslant |B_{R}|(J_{1}+\varepsilon),\\
    &\int_{B_{R}}\left|v_{*}\right|^{2}d\mathbf{r}\geqslant |B_{R}|(J_{2}-\varepsilon).
\end{aligned}
\end{equation}

Therefore, we obtain from (\ref{ine:s_r}) and (\ref{ine:R_e}) that
\begin{align*}
    \frac{\left\Vert \mathscr{H}_{\mathbf{\tilde{q}_{0}}}\psi_{R}-\tilde{\lambda}\psi_{R}\right\Vert_{L^{2}(\mathbb{R}^{d})}^{2}}{\left\Vert \psi_{R}\right\Vert_{L^{2}(\mathbb{R}^{d})}^{2}}\leqslant&\tilde{C}_{*}\frac{\int_{S_{R}}\left(\left|v_{*}\right|+\left|\nabla v_{*}\right|+\left|\Delta v_{*}\right|\right)^{2}d\mathbf{r}}{\int_{B_{R}}\left|v_{*}\right|^{2}d\mathbf{r}}\\=&\tilde{C}_{*}\frac{\left(\int_{B_{R+1}}-\int_{B_{R}}\right)\left(\left|v_{*}\right|+\left|\nabla v_{*}\right|+\left|\Delta v_{*}\right|\right)^{2}d\mathbf{r}}{\int_{B_{R}}\left|v_{*}\right|^{2}d\mathbf{r}}\\\leqslant&\tilde{C}_{*}\frac{\left(|B_{R+1}|-|B_{R}|\right)J_{1}+\varepsilon\left(|B_{R+1}|+|B_{R}|\right)}{ |B_{R}|(J_{2}-\varepsilon)}\rightarrow0
\end{align*}
as $R\rightarrow\infty$ and $\varepsilon\rightarrow0^{+}$, which together with Lemma \ref{lem:app_p} implies that $\tilde{\lambda}\in \sigma(\mathscr{H}_{\mathbf{\tilde{q}_{0}}})$.

We then turn to prove that
\begin{align*}
\tilde{\lambda}\in \sigma(\mathscr{H}).
\end{align*}

Define the translation flow on the torus $\mathbb{T}^{2d}$ by
\begin{align*}
    T_{\mathbf{t}}\tilde{\mathbf{r}}=(\mathbf{r}+\mathbf{t},\mathbf{r}'+\mathbf{t}),
\qquad \mathbf{t}\in\mathbb{R}^{d},
\end{align*}
for $\tilde{\mathbf{r}}=(\mathbf{r},\mathbf{r}')\in\mathbb{T}^{2d}.$ Note that $\mathcal{R}_{1}$ and $\mathcal{R}_{2}$ are incommensurate.
We see that, for each phase $\tilde{\mathbf{r}}\in\mathbb{T}^{2d}$, its orbit $\{T_{\mathbf{t}}\tilde{\mathbf{r}}\}_{\mathbf{t}\in\mathbb{R}^{d}}$ is dense in the whole phase space, i.e., 
\begin{align}\label{eq:bar_tor}
    \overline{\{T_{\mathbf{t}}\tilde{\mathbf{r}}:\mathbf{t}\in\mathbb{R}^{d}\}}
=\mathbb{T}^{2d}.
\end{align}

We claim that for $\mathbf{\tilde{q}_{0}}=(\mathbf{q_{0}}, -\mathbf{q_{0}})$, there holds 
\begin{align}\label{eq:spec_trans}
\sigma(\mathscr{H}_{T_{\mathbf{t}}\mathbf{\tilde{q}_{0}}})=\sigma(\mathscr{H}_{\mathbf{\tilde{q}_{0}}}), \quad \forall \mathbf{t}\in\mathbb{R}^{d}.
\end{align}
Given $\mathbf{t}\in\mathbb{R}^{d}$, define the translation operator $U_{\mathbf{t}}:L^2(\mathbb{R}^{d})\rightarrow L^2(\mathbb{R}^{d})$ by
\begin{align*}
    (U_{\mathbf{t}}f)(\mathbf{r})=f(\mathbf{r}+\mathbf{t}).
\end{align*}
We see that $U_{\mathbf{t}}$ is unitary for any $\mathbf{t}\in\mathbb{R}^{d}$. By some simple calculations, we have
\begin{align*} \mathscr{H}_{T_{\mathbf{t}}\mathbf{\tilde{q}_{0}}}=U_{\mathbf{t}}\mathscr{H}_{\mathbf{\tilde{q}_{0}}}U_{\mathbf{t}}^{-1}.
\end{align*}
Since unitary equivalent operators have identical spectra, we conclude $(\ref{eq:spec_trans})$ holds.

By (\ref{eq:bar_tor}), there exists a sequence $\mathbf{t}_n\in\mathbb{R}^{d}$ such that
\begin{align*}
T_{\mathbf{t}_n}\mathbf{\tilde{q}_{0}}\rightarrow\mathbf{0}
\end{align*}
in the torus $\mathbb{T}^{2d}$ as $n\rightarrow\infty$.
Equivalently, there hold that
\begin{align*}   \mathbf{q_{0}}+\mathbf{t}_n\rightarrow\mathbf{0},
\qquad
-\mathbf{q_{0}}+\mathbf{t}_n\rightarrow\mathbf{0}
\end{align*}
modulo the corresponding periods.
Consequently, due to the uniform continuity of the potential $V_{j} (j=1,2)$, we have
\begin{equation}
   \begin{aligned}\label{lim:poten}
    \left\Vert V_{1}(\cdot+\mathbf{q_{0}}+\mathbf{t}_n)-V_{1}(\cdot)\right\Vert_{L^{\infty}}&\rightarrow0,\\\left\Vert V_{2}(\cdot-\mathbf{q_{0}}+\mathbf{t}_n)-V_{2}(\cdot)\right\Vert_{L^{\infty}}&\rightarrow0
\end{aligned} 
\end{equation}
as $n\rightarrow\infty$.

Consider the resolvents 
\begin{align*}
    R(\mathrm{i},\mathscr{H})=(\mathscr{H}-\mathrm{i}I)^{-1},\quad  R(\mathrm{i},\mathscr{H}_{T_{\mathbf{t}_{n}}\mathbf{\tilde{q}_{0}}})=(\mathscr{H}_{T_{\mathbf{t}_{n}}\mathbf{\tilde{q}_{0}}}-\mathrm{i}I)^{-1}
\end{align*}
Following a similar argument to the proof in (\ref{ine:modh}), we obtain that 
\begin{align*}
    \left\Vert R(\mathrm{i}; \mathscr{H}) \right\Vert\leqslant1,\quad \left\Vert R(\mathrm{i}; \mathscr{H}_{T_{\mathbf{t}_{n}}\mathbf{\tilde{q}_{0}}}) \right\Vert\leqslant1
\end{align*}
Then, the resolvent identity and (\ref{lim:poten}) give
\begin{align*}
    \left\Vert R(\mathrm{i},\mathscr{H})-R(\mathrm{i},\mathscr{H}_{T_{\mathbf{t}_{n}}\mathbf{\tilde{q}_{0}}})\right\Vert=&\left\Vert R(\mathrm{i},\mathscr{H})\left(\mathscr{H}-\mathscr{H}_{T_{\mathbf{t}_{n}}\mathbf{\tilde{q}_{0}}}\right)R(\mathrm{i},\mathscr{H}_{T_{\mathbf{t}_{n}}\mathbf{\tilde{q}_{0}}})\right\Vert\\\leqslant&\left\Vert V_{1}(\cdot+\mathbf{q_{0}}+\mathbf{t}_n)-V_{1}(\cdot)\right\Vert_{L^{\infty}}+\left\Vert V_{2}(\cdot-\mathbf{q_{0}}+\mathbf{t}_n)-V_{2}(\cdot)\right\Vert_{L^{\infty}}\rightarrow0,
\end{align*}
as $n\rightarrow\infty$, which yields that
\begin{align*}
\mathscr{H}_{T_{\mathbf{t}_{n}}\mathbf{\tilde{q}_{0}}}\rightarrow \mathscr{H}
\quad\text{in norm resolvent sense}.
\end{align*}
Thus, the spectrum is upper semicontinuous with respect to this convergence \cite{de2008intermediate}. Namely,
\begin{align*}
    \limsup_{n\rightarrow\infty}
\sigma(\mathscr{H}_{T_{\mathbf{t}_{n}}\mathbf{\tilde{q}_{0}}})\subset\sigma(\mathscr{H}),
\end{align*}
which together with (\ref{eq:spec_trans}) yields that
\begin{align*}
    \tilde{\lambda}\in \sigma(\mathscr{H}_{\mathbf{\tilde{q}_{0}}})=\limsup_{n\rightarrow\infty}
\sigma(\mathscr{H}_{T_{\mathbf{t}_{n}}\mathbf{\tilde{q}_{0}}})\subset\sigma(\mathscr{H}).
\end{align*}
This completes the proof.
\end{proof}
\begin{remark}
 Theorem \ref{prop:susbs} tells that the spectrum $\sigma(\mathscr{H})$ is contained in the spectrum $\sigma_{\mathcal{G}}(\tilde{\mathscr{H}})$, and each spectral point in $\sigma_{\mathcal{G}}(\tilde{\mathscr{H}})$ satisfying the generalized eigenfunction condition is shown to correspond to a spectral point of the original operator $\mathscr{H}$.

    We mention that the generalized eigenfunction condition is a common assumption in the spectral
analysis of PDEs. 
 For instance, Bloch theory tells the existence of bounded Bloch solutions associated with every spectral point, and the Ishii--Pastur--Shnol theory \cite{damanik2022one} relates the spectrum of ergodic Schr\"odinger operators to generalized eigenfunctions with controlled growth. In our future work, we will explore the connection between the spectral points of the extended Schrödinger operator $\tilde{\mathscr{H}}$ after the self-adjoint extension and its  generalized eigenfunctions.
\end{remark}

The following conclusion follows from Corollary \ref{cor:dens} and Theorem \ref{prop:susbs}, which tells that the spectrum of the quasi-periodic Schrödinger operator $\mathscr{H}$ on the domain $H^{2}(\mathbb{R}^{d})$ can be approximated by the spectra of regularized Schrödinger operators $\left\{\tilde{\mathscr{H}}^{\delta}\right\}_{\delta>0}$ on the domain $H^{2}(\mathbb{R}^{2d})$, which are elliptic, retain periodicity, and enjoy favorable analytic and spectral properties as shown in Section 4. 
\begin{theorem}\label{thm:apppppp}
For $\lambda\in\sigma(\mathscr{H})$, there exists $\tilde{\lambda}^{\delta}\in\sigma(\tilde{\mathscr{H}}^{\delta})$ such that
\begin{align}\label{eq:lambdddd}
\lambda=\lim_{\delta\rightarrow0^{+}}\tilde{\lambda}^{\delta}.
\end{align}
\end{theorem}

\subsection{Electronic state distribution}

We see that for the regularized model (\ref{eq:re_ehn}) as a classical periodic system, the electronic state distribution is characterized by the probability density, which is defined as the squared modulus of the Bloch solutions. To describe the electronic state distribution in the incommensurate system (\ref{eq:incom_sys}), we can establish the definition of probability density via the asymptotic behavior of the regularized model.

Following \cite{MARTIN2020}, we introduce the definition of the set of probability densities for the systems (\ref{eq:incom_sys}) and (\ref{eq:re_ehn}).
\begin{definition}\label{def:sol_spa}
For any compact set $K\in\mathbb{R}^{d}$, 
for $\delta>0$ and $\tilde{\lambda}^{\delta}\in \sigma(\tilde{\mathscr{H}}^{\delta}) $, define 
\begin{align*}
     \Theta_{\delta}(\tilde{\lambda}^{\delta})=\left\{\left|\tilde{u}^{\delta}(\mathbf{r}, \mathbf{r})\right|^{2}: \text{given } \tilde{\lambda}^{\delta},\text{ }\tilde{u}^{\delta} \text{solves (\ref{eq:re_ehn}) and satisfies }\left\Vert \tilde{u}^{\delta}(\mathbf{r}, \mathbf{r})\right\Vert_{L^{2}(K)}=1 \right\}.
 \end{align*}
For $\lambda\in \sigma(\mathscr{H})$, define 
 \begin{align*}
\Theta_{0}(\lambda)
=
\Bigl\{
\varrho:\;
&\exists \{\delta_{\ell}\}_{\ell\geqslant1}\text{ and }\{\delta_{\ell_i}\}\subset \{\delta_{\ell}\} \text{ with }\;
\tilde{\lambda}^{\delta_{\ell_{i}}}\in\sigma(\tilde{\mathscr H}^{\delta_{\ell_{i}}}),
\;
\tilde{\lambda}^{\delta_{\ell_{i}}}\to\lambda,\text{ s.t. }\\
&|\tilde{u}^{\delta_{\ell_{i}}}(\mathbf r,\mathbf r)|^2
\in
\Theta_{\delta_{\ell_{i}}}(\tilde{\lambda}^{\delta_{\ell_{i}}}),
\text{ and }
\langle\varrho,\varphi\rangle
=
\lim_{i\to\infty}
\left\langle
|\tilde{u}^{\delta_{\ell_{i}}}(\mathbf r,\mathbf r)|^2,
\varphi
\right\rangle,
\quad
\forall\,\varphi\in C^\infty(K)
\Bigr\}.
\end{align*}
\end{definition}

The following theorem establishes the well-posedness of the probability density (in the distribution sense) of the incommensurate system (\ref{eq:incom_sys}), and shows that it can be approximated by the probability densities generated by the Bloch solutions of the regularized model (\ref{eq:re_ehn}).
\begin{theorem}\label{thm:solut}
For any compact set $K\in\mathbb{R}^{d}$ and  $\lambda\in\sigma(\mathscr{H})$ satisfying $\lambda=\lim_{\delta\rightarrow0^{+}}\tilde{\lambda}^{\delta}$ with $\tilde{\lambda}^{\delta}\in\sigma(\tilde{\mathscr{H}}^{\delta})$, there hold
\begin{enumerate}[label=(\alph*)]
    \item $\Theta_{0}(\lambda)\neq\{0\}$;
    \item \begin{align*}
\lim_{\delta \to 0^+} \operatorname{dist}_{K, \varphi}\bigl( \Theta_{\delta}(\tilde{\lambda}^{\delta}), \Theta_{0}(\lambda)\bigr) = 0, ~~~ \text{for} ~ \varphi \in C^{\infty}(K),
\end{align*}
where the distance from sets $X$ to $Y$ with respect to the compact set $K \subset \mathbb{R}^d$ and the test function $\varphi \in C^{\infty}(K)$ is given by
\begin{align*}
\operatorname{dist}_{K, \varphi}(X, Y) \coloneqq \sup_{x \in X} \inf_{y \in Y} \langle x - y, \varphi \rangle;
\end{align*}
\item if $\varrho\in\Theta_{0}(\lambda)\setminus\{0\}$ and $\varrho\in L^{1}(K)$, then there exists $\{\delta_{\ell_{i}}\}$ such that
\begin{align*}
\lim_{i\rightarrow\infty}\left\Vert\varrho-\left|\tilde{u}^{\delta_{\ell_{i}}}(\mathbf{r}, \mathbf{r})\right|^{2}\right\Vert_{L^{1}(K)}=0.
\end{align*}
\end{enumerate}

\end{theorem}
\begin{proof}
Given the compact set $K \subset \mathbb{R}^d$ and any sequence $\{\delta_{\ell}\}_{\ell\geqslant1}$, we define the Radon measure induced by the Bloch solutions as
\begin{align}
\label{eq:m_delta_L}
\mu^{\delta_{\ell}}:=\left|\tilde{u}^{\delta_{\ell}}(\mathbf{r}, \mathbf{r})\right|^{2}d\mathbf{r}.
\end{align}
 That is, for any $A \subset K$, there holds 
\begin{align*}
\mu^{\delta_{\ell}}(A) = \int_{A} \left|\tilde{u}^{\delta_{\ell}}(\mathbf{r}, \mathbf{r})\right|^{2}d\mathbf{r}.
\end{align*}

Note that $\mu^{\delta_{\ell}}(K)=1$ for $\ell\geqslant1$. We have that $\{\mu^{\delta_{\ell}}\}$ forms a bounded family in the measure space $\mathcal{M}(K)$. Since $K$ is compact, the Riesz representation theorem shows $\mathcal{M}(K) = C(K)^*$. Thus, by the Banach-Alaoglu Theorem, there exists some finite Radon measure $\mu$ and a subsequence $\{\delta_{\ell_i}\} \subset \{\delta_{\ell}\}$ such that $\mu^{\delta_{\ell_{i}}}$ weak-$*$ converges to $\mu$, i.e.,
\begin{align*}
\mu^{\delta_{\ell_{i}}}\rightharpoonup^* \mu, ~~~ \text{as} ~~ i \to \infty.
\end{align*}
That is, for $\varphi \in C^{\infty}(K)$, we have
\begin{align}
\label{eq:weak-star-convergence}
\int_K \varphi d \mu^{\delta_{\ell_{i}}} \rightarrow \int_K \varphi d\mu, ~~~ \text{as} ~~ i \to \infty.
\end{align}
This measure $\mu$ induces an probability density $\varrho$ in the sense of distributions, namely,
\begin{align}
\label{eq:weak-star-convergence-limit}
\inner{\varrho, \varphi}\coloneqq \int_{K} \varphi d \mu.
\end{align}
Notably, $\varrho\neq 0$ since if we choose the test function $\varphi  = 1$, then by (\ref{eq:weak-star-convergence}) and (\ref{eq:weak-star-convergence-limit}), we obtain that
\begin{align*}
\inner{\varrho, 1}= \mu(K) &= \lim_{i \to \infty} \mu^{\delta_{\ell_{i}}}(K)=1.
\end{align*}
Therefore, $\Theta_{0}(\lambda)\neq\{0\}$. Moreover, due to the arbitrariness of $\{\delta_{\ell}\}$,
there holds that for any test function $\varphi \in C^{\infty}(K)$,
\begin{align*}
\lim_{\delta \to 0^+} \operatorname{dist}_{K, \varphi}\bigl( \Theta_{\delta}(\tilde{\lambda}^{\delta}), \Theta_{0}(\lambda)\bigr) = 0.
\end{align*}

Consequently, if $\varrho\in\Theta_{0}(\lambda)$ and $\varrho\in L^{1}(K)$, we obtain from (\ref{eq:weak-star-convergence}) and (\ref{eq:weak-star-convergence-limit}) that 
\begin{align*}
\lim_{i\rightarrow\infty}\left\Vert\varrho-\left|\tilde{u}^{\delta_{\ell_{i}}}(\mathbf{r}, \mathbf{r})\right|^{2}\right\Vert_{L^{1}(K)}=0.
\end{align*}
This completes the proof.
\end{proof}

    Theorem \ref{thm:solut} establishes the existence and nontriviality of the limiting electronic state distribution in the incommensurate system, ensuring that the proposed quantity is well-defined. The analysis and numerical experiments in our forthcoming work \cite{ding2026regularized} will demonstrate that the induced  density can effectively characterize the properties of the system.

\section{Concluding remarks}
In this paper, we have tried to carry out the spectral analysis of the Schrödinger operator for the incommensurate system, which is a typical quasi-periodic system. We embed the system into higher dimensions and investigate the extended Schrödinger operator. We establish the spectral consistency between the extended Schrödinger operator after applying the self-adjoint extension and the original Schr\"odinger operator for the incommensurate system.

By introducing a regularization technique to deal with the degenerate ellipticity of the extended Schrödinger operator, we have obtained that the spectrum of the Schrödinger operator for the incommensurate system can be approximated by the spectra of a family of regularized Schrödinger operators, which are elliptic, retain periodicity, and enjoy favorable analytic and spectral properties. In addition, we have shown the
well-posedness of the probability density for the incommensurate system. The probability density can provide a description of the distribution of electronic
states of the system, and can be approximated by the
ones generated by the Bloch solutions to the regularized model. The regularized model is well-posed: all relevant physical observable quantities
are rigorously defined, and the classical Bloch theory, together with a broad range of efficient numerical algorithms developed for periodic systems, is applicable. Our analysis provides theoretical support and an efficient approach for the computation of observable quantities for the incommensurate systems, and can be extended to systems in any dimension and with any number of stacked material layers. It is our ongoing work to apply the regularized model based computational approach to quantum incommensurate systems \cite{ding2026regularized}.


	\bibliographystyle{siamplain}
	\bibliography{references}
\end{document}